\documentclass[
reprint, 
superscriptaddress, 
amsmath,amssymb, 
aps, 
prl, 
]{revtex4-2}

\usepackage{graphicx} 
\usepackage{dcolumn} 
\usepackage{bm} 

\usepackage{amsmath}

\usepackage{amsthm,amssymb, physics}
\usepackage{xcolor}
\usepackage{soul}
\usepackage{url}
\usepackage{enumitem}
\usepackage{appendix}
\usepackage{float}
\usepackage[hidelinks]{hyperref}
\usepackage{mathtools}
\usepackage{booktabs}
\usepackage{quantikz}
\usepackage{tikz}
\usepackage{makecell}
\usepackage[normalem]{ulem}
\usepackage{algorithm}
\usepackage{algpseudocode}
\usepackage{endnotes}

\newtheorem{Theorem}{Theorem}
\newtheorem{Lemma}[Theorem]{Lemma}

\newtheorem{definition}{Definition}

\newtheorem{corollary}{Corollary}

\newcommand{\HilbertSpace}[0]{\ensuremath{\mathcal{H}}}
\newcommand{\DensityMatrices}[0]{\ensuremath{\mathcal{D}(\HilbertSpace)}}

\newcommand{\KD}[0]{\ensuremath{Q}}
\newcommand{\POVM}[0]{\ensuremath{E}}
\newcommand{\Abasis}[0]{\ensuremath{\mathcal{A}}}
\newcommand{\Bbasis}[0]{\ensuremath{\mathcal{B}}}

\newcommand{\g}[0]{\ensuremath{\bm{g}}}
\renewcommand{\c}[0]{\ensuremath{\bm{\chi}}}
\newcommand{\h}[0]{\ensuremath{\bm{h}}}
\newcommand{\Pauli}[0]{\ensuremath{P}}

\newcommand{\Hadamard}[0]{\ensuremath{\mathrm{H}}}
\newcommand{\CX}[0]{\ensuremath{\mathrm{CNOT}}}
\newcommand{\G}[0]{\ensuremath{\mathbb{Z}_2^n}}
\newcommand{\e}[0]{\ensuremath{\bm{\eta}}}
\newcommand{\z}[0]{\ensuremath{\bm{0}}}
\newcommand{\range}[0]{\ensuremath{\textrm{ran\;}}}
\newcommand{\inputstate}[0]{\ensuremath{\rho_{\mathrm{input}}}}
\newcommand{\mana}[0]{\ensuremath{\mathcal{M}}}
\newcommand{\DGBR}[0]{\ensuremath{W}}
\newcommand{\ubold}[0]{\ensuremath{\bm{u}}}
\newcommand{\vbold}[0]{\ensuremath{\bm{v}}}
\newcommand{\xket}[1]{\ensuremath{\ket*{#1^{(x)}}}}
\newcommand{\zket}[1]{\ensuremath{\ket*{#1^{(z)}}}}
\newcommand{\xbra}[1]{\ensuremath{\bra*{#1^{(x)}}}}
\newcommand{\zbra}[1]{\ensuremath{\bra*{#1^{(z)}}}}
\newcommand{\zxbraket}[2]{\ensuremath{\braket*{#1^{(z)}}{#2^{(x)}}}}
\newcommand{\xzbraket}[2]{\ensuremath{\braket*{#1^{(x)}}{#2^{(z)}}}}
\newcolumntype{L}[1]{>{\raggedright\arraybackslash}p{#1}}

\begin{document}

\title{Kirkwood-Dirac Nonpositivity Is a Necessary Resource for Quantum Computing}

\author{Jonathan J. Thio}
\affiliation{Cavendish Lab., Department of Physics, Univ. of Cambridge, Cambridge CB3 0HE, UK}

\author{Songqinghao Yang}
\affiliation{Cavendish Lab., Department of Physics, Univ. of Cambridge, Cambridge CB3 0HE, UK}
\affiliation{Hitachi Cambridge Lab., J.J. Thomson Avenue, Cambridge CB3 0HE, UK}

\author{Nicole Yunger Halpern}
\affiliation{
Joint Center for Quantum Information and Computer Science, NIST and University of Maryland, College Park, Maryland 20742, USA
}
\affiliation{
Institute for Physical Science and Technology, University of Maryland, College Park, Maryland 20742, USA
}

\author{Stephan De Bièvre}
\affiliation{Univ. Lille, CNRS, Inria, UMR 8524, Laboratoire Paul Painlevé, F-59000 Lille, France}

\author{Crispin H. W. Barnes}
\affiliation{Cavendish Lab., Department of Physics, Univ. of Cambridge, Cambridge CB3 0HE, UK}

\author{David R. M. Arvidsson-Shukur}
\affiliation{Hitachi Cambridge Lab., J.J. Thomson Avenue, Cambridge CB3 0HE, UK}

\date{\today}

\begin{abstract}
    We elucidate the boundary between classical and quantum computation by constructing qubit Clifford circuits with nonstabilizer inputs that can be efficiently simulated classically. 
    We do so by casting 
    the quantum circuits realizable by defect braiding in the surface code 
    in terms of a Kirkwood-Dirac (KD) quasiprobability distribution, a generalization of a joint probability distribution. If this distribution remains a proper (positive) probability distribution throughout a circuit, then a classical algorithm can simulate the circuit efficiently. By leveraging recent results on the geometry of KD-positive states, we construct bound-magic states.
    Classical computers can efficiently simulate these bound-magic states' evolutions under the circuits, although other magic states enable universal quantum computation when inputted.
    Furthermore, we show that KD nonpositivity is a resource monotone in this model. Thus, we establish KD nonpositivity as a necessary resource for quantum-computational advantages.
\end{abstract}

\maketitle

\textit{Introduction:---}Quantum-computing research is motivated by the expectation that quantum computers can outperform their classical counterparts. However, identifying the source of quantum advantages is notoriously difficult \cite{Nielsen2010}. A deeper understanding of this issue is needed to advance quantum computing.

One may investigate the source of quantum advantages by comparing a universal model of quantum computation with a suitably restricted model. A common universal model is the quantum-computation-by-state-injection model \cite{Bravyi2005}. It entails arbitrary single-qubit state preparations, Clifford gates, and computational-basis measurements. One can restrict this model by constraining its input states. If classical computers can simulate the restricted model efficiently, then (if quantum advantages exist), quantum speedups require the excluded states.

The Gottesman-Knill theorem exemplifies this approach. The theorem governs a model that includes stabilizer-state preparations (defined below), certain gates (the Clifford gates), and computational-basis measurements. Classical algorithms can simulate the Gottesman-Knill model efficiently \cite{gottesman1998, Aaronson2004}. Consequently, \textit{magic} states---states that are not mixtures of stabilizer states---are necessary for quantum advantages \cite{Bravyi2005}. This result implies that some magic states enable quantum advantages in the quantum-computation-by-state-injection model. 
One can append certain magic states, called \emph{distillable-magic} states, to the Gottesman–Knill model. The enhanced model can efficiently simulate every quantum computation by state injection \cite{Bravyi2005, Reichardt2005}. However, not all magic is distillable. With respect to the Gottesman--Knill model, there may exist \textit{bound-magic} states---magic states that still yield a classically efficiently simulable model when added to the Gottesman--Knill model \cite{Campbell2010, Veitch2012}. We extend the set of known classically simulable quantum circuits by constructing bound-magic states for the computational model realized by defect braiding in the surface code. Although the Gottesman–Knill model includes our model, ours becomes universal if augmented by certain distillable-magic states \cite{Delfosse2015}.

\begin{figure}[t]
    \centering
        \begin{tikzpicture}[xscale=1.6, yscale=0.7, every node/.style={transform shape=false}]
        \draw (0,0) circle (2cm);
    
        \coordinate[label = below:\small{$\ket{c_1}$}] (a1) at (240:2);
        \coordinate[label = above :\small{$\ket{c_3}$}] (a2) at (90:2);
        \coordinate[label =  left:\small{$\ket{s}$}] (b1) at (170:2);
        \coordinate[label =  right:\small{$\ket{c_2}$}] (b2) at (340:2);
    
        \filldraw[fill=green!40!cyan!20, draw=black, dashed, thick] (a1) -- (b1) -- (a2);    

        \filldraw[red!40!yellow!20, opacity=1, draw=black, dotted, thick] 
          (a1)
          .. controls (0,-2) and (1,-1.5) .. (b2)  
          -- (b2) -- (a1);  

        \filldraw[red!40!yellow!20, opacity=1, draw=black, dotted, thick] 
          (a2)
          .. controls (0,2) and (1.5,1.5) .. (b2)  
          -- (b2) -- (a2);  
    
        \filldraw[fill=cyan!40!blue!30, draw=black,thick] (a1) -- (a2) -- (b2) -- (a1);

        \draw[dotted, thick] (a1).. controls  (-1,0) .. (a2);
    
        \node at (a1)[circle,fill,inner sep=1pt]{};
        \node at (a2)[circle,fill,inner sep=1pt]{};
        \node at (b1)[circle,fill,inner sep=1pt]{};
        \node at (b2)[circle,fill,inner sep=1pt]{};
        \node at (0,0)[circle,fill,inner sep=1pt]{};
        \draw[-stealth, thick] (0,0) -- (274:2);
        \node at (0.3,0){\small $I/d$};
        \node at (-0.1,-1) {\small $F$};
    \end{tikzpicture}
    \caption{\textbf{Schematic overview of the efficiently classically simulable states' geometry.} The disk represents the set of all two-qubit states. The boundary represents the pure states. The states $\ket{c_j}$ and $\ket{s}$ are CSS and non-CSS stabilizer states, respectively. The area bounded by the solid triangle represents the CSS polytope. The area bounded by the dashed lines represents further stabilizer mixtures. The shape enclosed by the dotted line represents the $\G$-KD-positive states. The yellow region represents bound-magic states relative to the DGBR model: they are $\G$-KD-positive but lie outside the stabilizer polytope.}
    \label{fig:geometryofstates}
\end{figure}

The main tools for identifying bound-magic states are quasiprobability distributions \cite{Veitch2012, Veitch2014resourcetheory}.  These mathematical objects resemble probability distributions but may assume (possibly nonreal) values outside $[0,1]$. 
Consider any quasiprobability distribution of a state. If the quasiprobabilities are probabilities (lie within $[0, 1]$), the state is called \emph{positive}.
Classical algorithms can efficiently simulate certain quantum computations whose states are always represented by positive quasiprobability distributions \cite{Veitch2012, Mari2012, Raussendorf2017, Park2024, Raussendorf2020, Okay2021, Zurel2020, Zurel2024, Zurel2025simulationquantumcomputationmagic, zurel2024efficientclassicalsimulationquantum, Pashayan2015}. In short, classical computers can efficiently simulate positive states.  Every positive state, if not a mixture of stabilizer states, is a bound-magic state.

A prominent quasiprobability-based simulation algorithm represents odd-dimensional qudits with the Gross-Wigner function \cite{Gross2006, Veitch2012}. This quasiprobability distribution's pure positive states are the pure stabilizer states. 
Moreover, one can use the Gross-Wigner function to identify bound-magic states defined on odd-dimensional Hilbert spaces \cite{footnote:GrossBoundStates}.
Since the Gross-Wigner function is defined only for odd-dimensional qudits, this construction cannot model qubit computations, the most common quantum computations \cite{SchmidWignerUniqueness}. References \cite{Delfosse2015, zurel2024efficientclassicalsimulationquantum,Zurel2020,Zurel2024,Zurel2025simulationquantumcomputationmagic, Zurel_2024} circumvented this limitation. In particular, Ref. \cite{Delfosse2015} introduced a quasiprobability distribution tailored to the operations implementable via defect braiding on the surface code. These operations act on \textit{rebits}.
A rebit is a qubit restricted to quantum states represented, relative to the computational basis, by real matrices \cite{rudolph2002}. We call this model the Delfosse-Guerin–Bian-Raussendorf (DGBR) model (Definition 1); and the associated quasiprobability distribution, the DGBR distribution. The DGBR distribution underlies an efficient classical-simulation algorithm for quantum circuits throughout which the state maintains a positive DGBR distribution \cite{Delfosse2015}.

In this Letter, we extend the volume of the known DGBR-positive rebit states. We do so by extending the DGBR distribution from rebits to all qubit states [Eqs. (2) and (3)]. The extended distribution, we prove, is the real part of a complex-valued Kirkwood-Dirac (KD) quasiprobability distribution \cite{Kirkwood33,Dirac45, Lostaglio2023kirkwooddirac, ArvidssonShukur2024} that we construct. We then show that the KD-positive states coincide with the DGBR-positive rebit states. 
Applying this equivalence, we characterize the rebit DGBR-positive states (Fig.~\ref{fig:geometryofstates}), using a recently developed toolkit for identifying KD-positive states~\cite{Langrenez2024, Langrenez2025Convex, debièvre2025}.
We thereby discover previously unknown bound-magic states, with respect to the DGBR model, defined on even-dimensional Hilbert spaces.
That is, we identify a family of quantum circuits that is universal, given suitable magic-state inputs, but efficiently classically simulable, given our bound-magic states.
Finally, we define the KD mana, a measure of our KD distribution's nonpositivity. The KD mana is an additive monotone for a resource theory of rebit quantum computation, which we define. Thus, we establish KD nonpositivity as a resource necessary for quantum-computational advantages. 

\textit{Preliminaries:---}Consider an $n$-qubit system associated with the Hilbert space $\HilbertSpace = (\mathbb{C}^{2})^{\otimes n}$ of dimensionality $d = 2^n$. Denote the set of density matrices on $\HilbertSpace$ by $\DensityMatrices$. Consider representing $\rho\in\DensityMatrices$ as a matrix relative to the computational basis. If and only if the elements are real, $\rho$ is \emph{real}, representing a rebit state. We index $d\times d$ matrices' elements by binary vectors $\ubold,\vbold \in \mathbb{Z}_2^n$. Furthermore, let $Z_j$ act as the Pauli-$Z$ operator on qubit $j$ and as the identity everywhere else. Define $X_j$ similarly. We denote a Pauli string associated with $\ubold,\vbold$ by
\begin{equation}\label{eq:real_Pauli_defn}
    \Pauli_{\ubold,\vbold} = \prod_{j=1}^nZ_j^{v_j} X_j^{u_j}.
\end{equation}

This Letter features stabilizer states, particularly Calderbank-Shor-Steane (CSS) states, defined as follows \cite{Calderbank98, Steane1996,gottesman1997}. The Pauli strings with phases $\pm1$ and $\pm i$ form the Pauli group. Consider any maximal commuting subgroup of the Pauli group. 
The subgroup's elements share exactly one eigenstate associated with the unit eigenvalue. The set of all such eigenstates, associated with all the maximal commuting subgroups, form the stabilizer states \cite{Nielsen2010}. The CSS states are the stabilizer states whose maximal stabilizing subgroup is generated by the $Z$ strings $\Pauli_{\z,\vbold}$ and $X$ strings $\Pauli_{\ubold,\z}$. The CSS states form the logical states of the popular CSS error-correction codes \cite{Calderbank98, Steane1996,gottesman1997}. 

Figure \ref{fig:geometryofstates} summarizes properties of the aforementioned states. The convex mixtures of the stabilizer states (the \emph{stabilizer mixtures}) form a polytope. Analogous statements govern the rebit-stabilizer states and the CSS states. The stabilizer polytope contains the rebit-stabilizer polytope, which contains the CSS polytope. Supplemental Note~I \cite{SuppMat} overviews all the states invoked in this Letter.

\textit{DGBR model and distribution:---}We now review an $n$-rebit computational model \cite{Delfosse2015}. This model's gates are the gates implementable via defect braiding in surface codes \cite{Raussendorf2007}. Such codes underlie several experimental research programs on fault tolerance
\cite{ExperimentalSurfaceCodes1, ExperimentalSurfaceCodes2, ExperimentalSurfaceCodes3}. Let $\Hadamard$ denote the Hadamard gate; and $\CX_{ct}$, the controlled-NOT gate controlled on qubit $c$ and targeting $t$.
\begin{definition}[DGBR model]
    The DGBR model includes the following operations:
    \begin{itemize}
        \item Initialization of qubits in CSS states. 
        \item Applications of $\Hadamard^{\otimes n}$, $\CX_{ct}$, and $\Pauli_{\ubold,\vbold}$.
        \item Measurement of any qubit in the computational basis.
    \end{itemize}
    Previous measurement outcomes may determine future operations adaptively.
\end{definition}

By the Gottesman--Knill theorem, classical computers can efficiently simulate the DGBR model. Furthermore, classical computers can efficiently simulate a DGBR model augmented with DGBR-positive states\cite{Delfosse2015}. The DGBR distribution $\DGBR: \DensityMatrices \rightarrow \mathbb{R}^{d\times d}$ is defined as
\begin{equation}\label{eq:DGBRinTermsOfA}
    \DGBR_{\ubold,\vbold}(\rho) = \Tr(A_{\ubold,\vbold}\rho) .
\end{equation}
The distribution's phase-space point operators  are
\begin{equation}\label{eq:DGBRphasespacepointoperators}
  A_{\z,\z} =  \sum_{\ubold\cdot\vbold  = 0 \hspace{-0.2cm}\mod 2} \frac{P_{\ubold,\vbold}}{2^{2n}}  \textrm{   and   }    A_{\ubold,\vbold} = P_{\ubold,\vbold} A_{\z,\z} P_{\ubold,\vbold}^{\dagger} . 
\end{equation}

A classical algorithm efficiently simulates the DGBR model enhanced by any DGBR-positive $\inputstate$ \cite{Delfosse2015}. The reason is that the DGBR distribution's positive pure rebit states are the CSS states, the distribution transforms covariantly under the DGBR model's unitary operations, and it satisfies a product rule with respect to tensor products.

Consider augmenting the DGBR model with the one- and two-rebit states. The augmented model is universal; it can efficiently simulate the quantum-computation-by-state-injection model and vice versa. By the quantum extended Church--Turing hypothesis \cite{Bernstein1997}, the augmented DGBR model therefore has the same computational power as any physical quantum computer. We may therefore study quantum computation using this model, without loss of generality.

\textit{The $\G$-Kirkwood-Dirac distribution:---}We now extend the DGBR distribution to a KD distribution, $\KD: \DensityMatrices \rightarrow \mathbb{C}^{d\times d}$ defined in terms of orthonormal bases $\Abasis$ and $\Bbasis$ for $\HilbertSpace$. ($\KD$ commonly represents KD distributions; this use of the symbol does not refer to the Husimi quasiprobability function \cite{Husimi1940}.) KD quasiprobability distributions \cite{Kirkwood33, Dirac45,ArvidssonShukur2024,Lostaglio2023kirkwooddirac} have recently been shown to have applications spanning quantum metrology \cite{Arvidsson-Shukur2020, lupu2022negative, Jenne22}, quantum foundations \cite{Aharonov88,Dressel14, Pang14, Pusey14, Dressel15, Pang15, Thio2024, Langrenez2025Convex}, quantum thermodynamics \cite{YungHalp17, Lostaglio18, Lostaglio20,Levy20,lostaglio2023kirkwood, Upadhyaya2024}, and information scrambling \cite{YungHalp17,YungHalp18,Razieh19,YungHalp19, Yunger19-2, Gherardini24}.

A KD distribution's bases can be tailored to a problem. To facilitate our quantum-computational analysis, we choose for $\Abasis$ to consist of the tensor products of the Pauli-$Z$ operator's eigenstates. We define $\Bbasis$ analogously in terms of the Pauli-$X$ operator.
    The resulting KD distribution is an instance of the general KD distribution for a finite Abelian group $G$  \cite{debièvre2025}. This group structure underlies our results. We have chosen for $G$ to be $\G$, which leads to the aforementioned $\Abasis$ and $\Bbasis$ choices. The group consists of elements $\g = (g_1, g_2, \dots, g_n)$, where $g_j \in \{0,1\}$. The group operation is addition modulo $2$ ($+$). The group elements label the computational basis, $\Abasis = \{\zket{\g}\}_{\g\in G}$. The superscript $(z)$ denotes the Pauli-$Z$ basis. Later, $(x)$ denotes the Pauli-$X$ basis. The group $G$ has a character group $\hat{G}$. We denote also its group operation by $+$. $G$ and $\hat{G}$ parallel a continuous-variable system's position and momentum spaces. Momentum is a vector and a translation operator. Similarly, a character $\c \in \hat{G}$  is a vector $\c = (\chi_1, \chi_2,\dots,\chi_n)$, where $\chi_k \in \{ 0,1\}$, and a function on $G$: $\c(\g) = \prod_{j=1}^n (-1)^{\chi_j{g_j}}$. Using the characters, we define the states
\begin{equation}
    \xket{\c} \coloneqq \frac{1}{|G|^{1/2}}\sum_{\g\in G}\c(\g)\zket{\g} .
\end{equation}
 $|G|$ denotes the number of elements in $G$. The second orthonormal basis is $\Bbasis = \{\xket{\c}\}_{\c \in \hat{G}}$.

With respect to this $\Abasis$ and $\Bbasis$, a state $\rho$ is represented by the KD distribution
\begin{equation}\label{eq:Z/2Z_KD_defn}
    \KD_{\g,\c}(\rho) \coloneqq \Tr(B_{\g,\c}\rho).
\end{equation}
The $\G$-KD distribution has the phase-point operators
\begin{equation}
\label{eq:KDPhaseSpaceOps}
B_{\g,\c} = \xket{\c} \xzbraket{\c}{\g}{\zbra{\g}} .
\end{equation}
Since $\zxbraket{\g}{\c} \neq 0$, this KD distribution is informationally complete on $\mathcal{H}$ \cite{ArvidssonShukur2024}, unlike the DGBR distribution. Moreover, since $\zxbraket{\g}{\c}$ is real, $\rho$ is a rebit state if and only if $\KD_{\g,\c}(\rho)$ is real. Hence, all KD-positive states are rebit states. 
Like $\rho$, measurements have KD representations \cite{ArvidssonShukur2024}. Let $\{ \POVM_l \}_l$ denote a set of positive-semidefinite matrices that represents a generalized measurement \cite{Nielsen2010}. The KD representation of $\{ \POVM_l \}_l$ is
\begin{equation}\label{eq:KD_character}
    \tilde{\KD}_{\g,\c}(\POVM_l) = \frac{ \Tr(B_{\g,\c}\POVM_l)}{|{\zxbraket{\g}{\c}}|^2}   .
\end{equation}
A time-evolving state's KD distribution, with a generalized measurement's KD representation, models an experiment \cite{Arvidsson-Shukur2020, ArvidssonShukur2024}. The outcome probabilities follow from the overlap formula \cite{ArvidssonShukur2024, debièvre2025}
\begin{equation}\label{eq:overlap_formula}
    \Tr(\POVM_l \rho) = \sum_{\g,\c} \tilde{\KD}_{\g,\c}(\POVM_l) \KD_{\g,\c}(\rho).
\end{equation}

We now show that the $\G$-KD distribution generalizes the DGBR distribution and enjoys similar simulation properties.
\begin{Theorem}\label{thm:properties}
    The $\G$-KD distribution has the following properties:
    \begin{enumerate}
        \item (DGBR connection) The real part of the $\G$-KD distribution \cite{Endnote1} equals the DGBR distribution: 
        \begin{equation}\label{eq:DGBRconnection}
        \Re\left[\KD_{\g,\c}(\rho)\right] = W_{\g,\c}(\rho).
        \end{equation}
        \item (Hudson's theorem) Any given pure state is $\G$-KD-positive if and only if it is a CSS state.
                \item (Covariance) The $\G$-KD distribution is covariant with respect to the DGBR operations
        \begin{align}\label{eq:covariance_of_paulis}
            \KD_{\g_0,\c_0}(\Pauli_{\g,\c}\rho \Pauli_{\g,\c}^{\dagger}) &= \KD_{\g_0 + \g,\c_0 + \c}(\rho),\\ \label{eq:covariance_of_Hadamards}
            \KD_{\g,\c}(\Hadamard^{\otimes n}\rho \Hadamard^{\otimes n\dagger}) &= \overline{\KD_{\c,\g}(\rho)}, \textrm{ and} \\\label{eq:covariance_of_CNOTs}
            \KD_{\g,\c}(\CX_{ct}\rho \CX_{ct}^\dagger) &= \KD_{A_{ct}\g, B_{ct}\c} (\rho) .
        \end{align}
        The overline denotes complex conjugation,  $(A_{ct})_{jk} \coloneqq \delta_{jk} + \delta_{tj}\delta_{ck}$, and $(B_{ct})_{jk} \coloneqq \delta_{jk} + \delta_{cj}\delta_{tk}$.
        \item (Products) Let $\HilbertSpace$ factorize as the tensor product $\HilbertSpace = \HilbertSpace_1\otimes\HilbertSpace_2$. Let $\rho_1 \in \mathcal{D}(\mathcal{H}_1)$ and $\rho_2 \in \mathcal{D}(\mathcal{H}_2)$. 
        The $\G$-KD distribution factorizes as
        \begin{equation}
            \KD_{(\g,\g'),(\c,\c')}(\rho_1\otimes\rho_2)= \KD_{\g,\c}(\rho_1)\KD_{\g',\c'}(\rho_2).
        \end{equation}
    \end{enumerate}
\end{Theorem}

To prove Theorem \ref{thm:properties}.1, we interrelate the DGBR distribution's phase-point operators [Eq. \eqref{eq:DGBRphasespacepointoperators}] and the $\G$-KD distribution's [Eq. \eqref{eq:KDPhaseSpaceOps}]:
\begin{equation}\label{eq:phase_space_point_operator_connection}
	A_{\g,\c} = \frac{1}{2} (B_{\g,\c} + B_{\g,\c}^{\dagger}).
\end{equation}
Substituting Eq. \eqref{eq:phase_space_point_operator_connection} into Eq. \eqref{eq:DGBRinTermsOfA} yields Theorem \ref{thm:properties}.1 (Supplemental Note~II \cite{SuppMat}). By Theorem \ref{thm:properties}.1, the DGBR-positive rebit states are the KD-positive states.

Hudson's original theorem \cite{Hudson1974} identifies the continuous-variable Wigner function's pure positive states \cite{Wigner1932}. Our Hudson theorem's proof builds on a recent result \cite{debièvre2025}. Consider the KD distributions based on any finite Abelian group $G$. All such distributions' pure KD-positive states have the form 
\begin{equation}\label{eq:KD_positive_pure_states}
    \ket{H; \g, \c}\coloneqq \frac{1}{|H|^{1/2}}\Pauli_{\g,\c}\sum_{\h\in H}\zket{\h}.
\end{equation}
$H$ denotes any subgroup of $G$. We compute the maximal stabilizing group of the states in Eq.~\eqref{eq:KD_positive_pure_states} and that of the CSS states (Supplemental Note~III \cite{SuppMat}). These groups are identical. Thus, so are the sets of states.

The covariance relation \eqref{eq:covariance_of_paulis} follows from Eq. (29) in Ref. \cite{debièvre2025}. We prove covariance relationships
\eqref{eq:covariance_of_Hadamards} and \eqref{eq:covariance_of_CNOTs} by evaluating the left-hand sides and using the fact that $\Hadamard$ and $\CX$ map eigenstates of real
Pauli operators to eigenstates of real Pauli operators (Supplemental Note~IV \cite{SuppMat}).
Theorem \ref{thm:properties}.4 follows from direct calculation.

\textit{Efficient classical simulation and bound-magic states:---}We now identify KD-positive states that lie outside the stabilizer polytope and therefore have bound magic relative to the DGBR model (Fig. \ref{fig:geometryofstates}). Theorem \ref{thm:properties}.1 implies that every KD-positive state is DGBR-positive. As DGBR-positive input states are efficiently classically simulable, so is the DGBR model enhanced by any $\G$-KD-positive input state $\inputstate$. We rephrase the DGBR simulation algorithm in terms of KD distributions in Supplemental Note~V \cite{SuppMat}.

To identify bound-magic states, we use tools for characterizing KD-positive states \cite{Langrenez2024, langrenez2024setkirkwooddiracpositivestates, debièvre2025}. 
The single-qubit $\mathbb{Z}_2$-KD-positive states are the mixtures of $\zket{\g}\zbra{\g}$ and $\xket{\c}\xbra{\c}$ (Theorem 1.1.i of Ref. \cite{Langrenez2024}). These states are the single-qubit rebit-stabilizer states. 
Furthermore, the DGBR algorithm can efficiently simulate only $\mathbb{Z}_2$-KD-positive states. Therefore, if the input is a tensor product of single-rebit states, the DGBR algorithm can simulate only DGBR circuits on rebit stabilizer mixtures.
Thus, the $\G$-KD distribution cannot identify single-qubit bound magic states within the DGBR model.

Nevertheless, the set of KD-positive states has a richer structure in higher dimensions: it may include \textit{exotic} KD-positive mixed states that are not mixtures of KD-positive pure states \cite{Langrenez2024, debièvre2025}. We now show that some two-qubit KD-positive mixed states are not mixtures of stabilizer states. 
Informed by the bound-magic states' locations in Fig. \ref{fig:geometryofstates}, we define an operator $F$ that is orthogonal to a facet of the CSS polytope and to a real ridge of the stabilizer polytope (which is a facet of the rebit-stabilizer polytope):
\begin{equation}
    F = 
    \begin{pmatrix}
        1         &  0          &  1            &  1         \\
        0         &  1          & -1            & -1         \\
        1         & -1          & -1            & -2         \\
        1         & -1          & -2            & -1
\end{pmatrix}.
\end{equation}
Using techniques from Refs. \cite{debièvre2025, Langrenez2024}, we parametrize with $\lambda \in [- 1/(-4 + 8 \sqrt{2}), \, 1/(4 + 8 \sqrt{2}) ]$ a state $\rho_{\lambda}$ along the direction of $F$: 
\begin{equation}
    \rho_{\lambda} = \frac{1}{4}I + \lambda F.
\end{equation}

Every state $\rho$ on the CSS-and-rebit-stabilizer facet's inner side satisfies $\Tr (\rho F) \leq 1$. In Supplemental Note~VI \cite{SuppMat}, we show that $\rho_\lambda$ is KD-positive if $\lambda\in [0, 1/(4+8\sqrt{2})]$. 
If $\lambda > 1/20$, then $\Tr(F\rho_\lambda) > 1$, so $\rho_\lambda$ lies outside the facet.
 Thus, if $\lambda \in ( 1/20, 1/(4+8\sqrt{2})]$, $\rho_\lambda$ is a two-qubit bound-magic state relative to the DGBR model. Similar bound-magic states exist above each of the 20 facets shared by the CSS polytope and the rebit-stabilizer polytope. We prove this claim in Supplemental Note~VII \cite{SuppMat}, using cdd \cite{Fukuda1995}, a numerical package for analyzing polytopes \cite{NISTDisclaimer}. The code used is available from Ref. \cite{Thi0_2025_KDQuantumComputation}.
 
Since $2$-qubit bound-magic states exist, bound-magic states must exist for every number $>2$ of qubits. The reason is that partially tracing over any stabilizer mixture generates another stabilizer mixture \cite{Audenaert2005}. To illustrate, we denote by $\sigma$ an arbitrary KD-positive single-qubit state. Let $\lambda\in (1/20, 1/(4+8\sqrt{2})]$.
$\sigma^{\otimes (n-2)}\otimes\rho_\lambda$ is KD-positive by Theorem \ref{thm:properties}.4. Moreover, $\sigma^{\otimes (n-2)}\otimes\rho_\lambda$ cannot be a stabilizer mixture: suppose it were a stabilizer mixture, and consider tracing out the first $n-2$ qubits. The result, $\rho_\lambda$, would be a stabilizer mixture \cite{Audenaert2005}; but it is not, by construction.
 Hence bound-magic states exist for every number $>2$ of qubits. A similar statement is true for exotic KD-positive states. Since the pure KD-positive states' convex hull is the CSS polytope, and the CSS polytope is a subset of the stabilizer polytope, this also implies that exotic states exist for every number of qubits greater than $2$.

Now, we quantify how much the DGBR simulation algorithm differs from the Gottesman--Knill simulation algorithm on DGBR circuits. 
We numerically compare the (Hilbert-Schmidt) volumes efficiently simulable by the simulation methods (Table~\ref{volume}). Approximately $0.69$ \% of the 2-rebit states are KD-positive bound-magic states. This result increases the known set of efficiently classically simulable states by $15$ \%. Supplemental Note~VIII \cite{SuppMat} details our sampling methodology. 

\begin{table}[hbt!]
    \centering
    \begin{tabular}{@{}lc@{}lc@{}}
        \toprule
        \textbf{Category} & \textbf{Eff. sim. by}  & \textbf{ Proportion } \\
        \midrule
        Stabilizer \& KD-positive   & GK\&DGBR    & $1.5614(5)$ \%  \\
        Stabilizer \& KD-nonpositive   & GK        & $2.9753(5)$ \%  \\
        Magic \& KD-positive (bound)& DGBR        & $0.6868(3)$ \%   \\
        Magic \& KD-nonpositive        & Neither      & $94.7766(6)$ \%  \\
        \bottomrule
    \end{tabular}
    \caption{\textbf{Estimated relative volumes of the $2$-qubit quantum states.}  We estimated the relative volumes by sampling $1$ billion random quantum states from the space $\mathcal{D} \bm{(}(\mathbb{R}^2)^{\otimes 2} \bm{)}$. (GK stands for Gottesman--Knill.)}
\label{volume}
\end{table}

We expect the bound-magic states' relative volume to vanish as $n\to\infty$. This conjecture does not prevent our algorithm from simulating reasonable DGBR models augmented with $\G$-KD-positive states. In such models, one injects tensor products of $k$-qubit states, with $k$ bounded by a fixed $n$-independent number. Allowing $k$ to be unbounded would eliminate the need for computation, trivializing the model, as one could simply initialize the final state. Hence, our method reveals a finite relative volume of bound-magic states.

\textit{KD nonpositivity as a resource necessary for quantum advantages:---}Next, we bound the KD nonpositivity needed for input-enhanced DGBR models to efficiently simulate universal quantum computers. We construct a resource theory of the DGBR model with arbitrary input states. The free states are the CSS mixtures, and the free operations are those in the DGBR model. We define the \emph{KD mana} as
\begin{equation}\label{eq:mana}
    \mana(\rho) \coloneqq \log \sum_{\g,\c}|\KD_{\g,\c}(\rho)| .
\end{equation}
$\mana(\rho)$ quantifies the nonpositivity in $\KD_{\g,\c}(\rho)$.
\begin{Theorem}\label{thm:properties_of_mana} The KD mana has the following properties:
    \begin{enumerate}
        \item (Faithfulness) $\mana(\rho) = 0$ if and only if $\rho$ is KD-positive.
        \item (Additivity) For all $n$-qubit density operators $\rho$ and $\sigma$, 
        \begin{equation}
            \mana(\rho\otimes\sigma) = \mana(\rho) + \mana(\sigma) .
        \end{equation}
        \item (Monotone) $\mana$ is nonincreasing under the free operations.
    \end{enumerate}
\end{Theorem}
Theorem \ref{thm:properties_of_mana}.1 is true because $\sum_{\g,\c}|\KD_{\g,\c}(\rho)|=1$ if and only if $\rho$ is $\G$-KD-positive \cite{ArvidssonShukur2024}. Theorem \ref{thm:properties_of_mana}.2 follows from the product rule in Theorem \ref{thm:properties}.4. We prove Theorem \ref{thm:properties_of_mana}.3 in Supplemental Note~IX \cite{SuppMat}, by explicit calculation.

Using the resource theory, we can quantify the KD nonpositivity necessary for quantum computation. A DGBR model augmented by the states
\begin{align}
    \ket{A} &= \frac{\ket{0}\ket{0} + \cos \frac{\pi}{4}\ket{1}\ket{0} + \sin \frac{\pi}{4} \ket{0}\ket{1}}{\sqrt{2}} \; \; \text{, and} \\
    \ket{B} &= \frac{\ket{0}\ket{+}+\ket{1}\ket{-}}{\sqrt{2}}.
\end{align}
can perform universal quantum computation \cite{Delfosse2015}. Consider starting with $m$ copies of a given state $\rho$ and transforming them, via free operations, into copies of the desired state $\sigma$. We lower-bound the required $m$:

\begin{corollary}
    Consider any procedure that has a nonzero probability of mapping $\rho^{\otimes m}$ to $\sigma^{\otimes n}$. The procedure requires $m > n \mana(\sigma)/\mana(\rho)$ copies of $\rho$, on average over trials.
\end{corollary}
The states $\ket{A}$ and $\ket{B}$ have nonzero KD mana [$\mana(\ket{A}\bra{A})$=0.473 and $\mana(\ket{B}\bra{B})$=0.870]. Since KD mana quantifies KD nonpositivity, we have established KD nonpositivity as a resource necessary for quantum computation.

\textit{Discussion:---}We have identified methods for quantifying quantum-computational advantages within the DGBR model enhanced with arbitrary input states. We generalized the rebit-restricted DGBR distribution to a $\G$-KD distribution.  Adapting the KD toolkit, we identified states that, when appended to the DGBR model, yield an efficiently classically simulable model. We thereby construct bound-magic states on even-dimensional Hilbert spaces. They increase by $\approx 15$ $\%$ the volume of the known efficiently classically simulable states.

Our construction answers and extends beyond an open question: do even-dimensional Hilbert spaces support bound-magic states? Campbell \textit{et al.} \cite{Campbell2010} showed that standard protocols cannot distill all qubit magic states. The proof does not govern all distillation protocols, however. Furthermore, researchers had not constructed bound-magic states using the DGBR distribution (unlike the Gross-Wigner function in the context of odd-dimensional Hilbert spaces). 

Our method simulates DGBR circuits with bound-magic-state inputs, although earlier methods (including Gottesman--Knill) cannot. We have thus enlarged the set of circuits known to be classically efficiently simulable. This result is independent of the DGBR model.

Last, we defined the KD mana, which quantifies the $\G$-KD distribution's nonpositivity. The KD mana, we showed, is a resource monotone for DGBR computation with arbitrary input states. The KD mana provides an easy-to-calculate bound on the rate of state conversion. Our work shows that KD nonpositivity is a necessary resource for quantum advantages in the DGBR model enhanced by arbitrary input states.

The $\G$-KD distribution is one instance of a broader class of KD distributions~\cite{debièvre2025}. Other instances simulate different computational models. The KD distribution is therefore not DGBR-specific, but a general tool for studying classical simulability. A natural next step combines our algorithm with the dynamically-updated-quasiprobability scheme of Park \textit{et al.}~\cite{Park2024}, potentially enlarging the bound-magic set or reaching nonreal states.

\medskip

\textit{Acknowledgements:---}JJT was supported by the Cambridge Trust. The authors thank Christopher Long, Wilfred Salmon, Joe Smith, Richard Jozsa, Nicolas Delfosse, Ryuji Takagi, and Cory Aitchison for helpful discussions. This work received support from the National Science Foundation: QLCI grant OMA-2120757. This work was supported in part by the Agence Nationale de la Recherche under grant ANR-11-LABX0007-01 (Labex CEMPI), by the Nord-Pas de Calais Regional Council and the European Regional Development Fund through the Contrat de Projets \'Etat-R\'egion (CPER), and by the CNRS through the MITI interdisciplinary programs.

\clearpage
\onecolumngrid
\begin{center}
  \textbf{\large Supplemental Material:\\ Kirkwood-Dirac Nonpositivity Is a Necessary Resource for Quantum Computing}\\[0.4cm]
  Jonathan J. Thio,$^{1}$ Songqinghao Yang,$^{1,2}$ Nicole Yunger Halpern,$^{3,4}$\\[0.1cm]
  Stephan De Bièvre,$^{5}$ Crispin H. W. Barnes,$^{1}$ and David R. M. Arvidsson-Shukur$^{2}$\\[0.35cm]
  {\itshape\small
  $^{1}$Cavendish Lab., Department of Physics, Univ. of Cambridge, Cambridge CB3 0HE, UK\\
  $^{2}$Hitachi Cambridge Lab., J.J. Thomson Avenue, Cambridge CB3 0HE, UK\\
  $^{3}$Joint Center for Quantum Information and Computer Science, NIST and University of Maryland, College Park, Maryland 20742, USA\\
  $^{4}$Institute for Physical Science and Technology, University of Maryland, College Park, Maryland 20742, USA\\
  $^{5}$Univ. Lille, CNRS, Inria, UMR 8524, Laboratoire Paul Painlevé, F-59000 Lille, France}\\[0.35cm]
  (Dated: \today)
\end{center}


\setcounter{equation}{0}
\setcounter{figure}{0}
\setcounter{table}{0}
\setcounter{page}{1}
\renewcommand{\theequation}{S\arabic{equation}}
\setcounter{secnumdepth}{3}
\setcounter{section}{0}
\renewcommand{\thesection}{\Roman{section}}

\section{Overview of Sets of States}
In this Supplemental Note, we provide an overview with of the various sets of states used in the main article. These sets are all subsets of the Hilbert space $\mathcal{H} = (\mathbb{C}^2)^{\otimes n}$ associated with $n$-qubit systems.

\begin{table}[hbt!]
\centering
\newcommand{\tA}[1]{\parbox[t]{0.15\linewidth}{\raggedright #1\smallskip}}
\newcommand{\tB}[1]{\parbox[t]{0.45\linewidth}{\raggedright #1\smallskip}}
\newcommand{\tC}[1]{\parbox[t]{0.30\linewidth}{\raggedright #1\smallskip}}
\begin{tabular}{lll}
\toprule
\textbf{Name} & \textbf{Description} & \textbf{Notes} \\
\midrule
\tA{Stabilizer states} & \tB{Pure states that are described by stabilizing subgroups.} & \tC{See Ref.~\cite{Nielsen2010} for more details.} \\
\tA{Stabilizer polytope} & \tB{Set of convex mixtures of stabilizer states.} & \tC{-} \\
\tA{Magic states} & \tB{Quantum states that lie outside the stabilizer polytope.} & \tC{See Ref.~\cite{Bravyi2005} for more details.} \\
\tA{Rebit states} & \tB{Quantum states with real density matrices in the computational basis.} & \tC{See Ref.~\cite{rudolph2002} for more details.} \\
\tA{Rebit stabilizer polytope} & \tB{Set of states that lie inside the stabilizer polytope and are also rebit states} & \tC{-} \\
\tA{CSS states} & \tB{Stabilizer states with stabilizing subgroups generated Pauli-$Z$ and Pauli-$X$ strings.} & \tC{See Ref.~\cite{Calderbank98, Steane1996,gottesman1997} for more details.} \\
\tA{CSS polytope} & \tB{Set of convex mixtures of CSS states.} & \tC{-} \\
\tA{$\G$-KD-positive states} & \tB{Quantum states such that the entries of the corresponding $\G$-KD distribution all lie in $[0,1]$.} & \tC{The pure $\G$-KD-positive states are the CSS states (Theorem 1.3 of the main article). However, there are mixed $\G$-KD-positive states that lie outside the CSS-polytope (Equation (17) of the main article).} \\
\tA{DGBR-positive states} & \tB{Quantum states such that the entries of the corresponding DGBR distribution all lie in $[0,1]$.} & \tC{The $\G$-KD-positive rebit states coincide with the DGBR-positive rebit states (Theorem 1.1 of the main article).} \\
\tA{Bound-magic states with respect to the DGBR model} & \tB{Magic states that still yield a classically efficiently simulable model when appended to the DGBR model.} & \tC{$\G$-KD-positive states that lie outside the stabilizer polytope are bound-magic states with respect to the DGBR model} \\
\tA{Exotic states} & \tB{Mixed KD-positive states that cannot be decomposed into only pure KD-positive states.} & \tC{The main article shows that such states must always exist on systems with 2 qubits or more.} \\
\bottomrule
\end{tabular}
\caption{Overview of Sets of States.}
\end{table}

\section{Proof of the DGBR connection}
We prove Theorem 1.1 of the main article in this Supplemental Note. The phase-space point operators of the $\G$-KD distribution are
\begin{equation}
    B_{\g,\c} = \xket{\c} {\xzbraket{\c}{\g}}{\zbra{\g}}.
\end{equation}

We start by showing that the phase-space point operators of the $\G$-KD distribution may be expressed in terms of the phase-space point operators of the DGBR distribution.
\begin{Lemma} The phase-space point operators of the DGBR distribution are the hermitianizations of the phase-space point operators of the $\G$-KD distribution:
    \begin{equation}\label{eq:hermitianization_equals}
        A_{\g,\c} = \frac{1}{2} (B_{\g,\c} + B_{\g,\c}^{\dagger}).
    \end{equation}
\end{Lemma}
\begin{proof}
        Using that $A_{\g,\c} = \Pauli_{\g,\c} A_{\z,\z} \Pauli_{\g,\c}^{\dagger}$, one may rewrite Eq. \eqref{eq:hermitianization_equals} as
    \begin{equation}\label{eq:A00intermsofB}
        \frac{1}{2}\left(\Pauli_{\g,\c}^\dagger B_{\g,\c} \Pauli_{\g,\c} + (\Pauli_{\g,\c}^\dagger B_{\g,\c} \Pauli_{\g,\c})^{\dagger}\right) = A_{\z,\z}.
    \end{equation}
    One may calculate the two terms on the left-hand side to be
    \begin{align}
        \Pauli_{\g,\c}^\dagger B_{\g,\c} \Pauli_{\g,\c} &= \prod_{j=1}^n X_j^{g_j} Z_j^{\chi_j}\xket{\c}\xzbraket{\c}{\g}\zbra{\g} \prod_{k=1}^nZ_k^{\chi_k} X_k^{g_k} \nonumber\\
        &= \xzbraket{\c}{\g} (-1)^{\c\cdot\g}\xket{(\c+\c)}\zbra{(\g+\g)} \nonumber\\
        &= \frac{1}{2^{n/2}}\ket{++\dots+}\bra{00\dots0},\label{eq:PBPexpanded}
    \end{align}
    where we used that $X$ and $Z$ act as boost and shift operators on $\zket{\g}$ and $\xket{\c}$ to go from the first to the second line, and we used that $\c$ and $\g$ are their own inverses to go from the second to the third line. Combining Eqs. \eqref{eq:A00intermsofB} and \eqref{eq:PBPexpanded}, we find that the proof amounts to showing that
    \begin{equation}\label{eq:A00simplified}
        A_{\z,\z} = \frac{1}{2\cdot2^{n/2}}\left[\ket{++\dots+}\bra{00\dots0} + (\ket{++\dots+}\bra{00\dots0})^\dagger\right].
    \end{equation}
    We show this by expanding $A_{\z,\z}$:
    \begin{align}
        A_{\z,\z} &= \frac{1}{2^{2n}} \sum_{\substack{\g,\c : \\\g\cdot\c = 0 \mod 2}}\Pauli_{\g,\c}\nonumber\\
        &= \frac{1}{2\cdot2^{2n}}\left(\sum_{\g,\c}\Pauli_{\g,\c} + \sum_{\g,\c}(-1)^{\g\cdot\c}\Pauli_{\g,\c}\right)\nonumber\\
        &= \frac{1}{2\cdot2^{2n}}\left(\sum_{\g,\c}\Pauli_{\g,\c} + \left[\sum_{\g,\c}\Pauli_{\g,\c}\right]^{\dagger}\right)\label{eq:A00symmetrized}
    \end{align}
    Comparing Eqs. \eqref{eq:A00simplified} and Eqs. \eqref{eq:A00symmetrized}, we find that it suffices to show that 
    \begin{equation}\label{eq:twooperatorsneedingtobeequal}
        \frac{1}{2^{n/2}}\ket{00\dots0}\bra{++\dots+} = \frac{1}{2^{2n}}\sum_{\g,\c}\Pauli_{\g,\c}.
    \end{equation}
    We do this by checking the matrix elements of the right-hand side:
    \begin{align}
        \frac{1}{2^{2n}}\zbra{\h}\sum_{\g,\c}\Pauli_{\g,\c}\xket{\e} &= \frac{1}{2^{2n}}\sum_{\g,\c} \zbra{\h} \prod_{j}Z_{j}^{\chi_j} X_{j}^{g_j} \xket{\e}\nonumber\\
        &= \frac{1}{2^{2n}}\sum_{\g,\c} (-1)^{\g\cdot\c} \zxbraket{(\g+\h)}{(\c+\e)}\nonumber\\
        &= \frac{(-1)^{\h\cdot\e}}{2^{n/2}2^{2n}}\sum_{\g} (-1)^{\g\cdot\e}\sum_{\c}(-1)^{\h\cdot\c}\nonumber\\
        &= \frac{1}{2^{n/2}}\delta_{\h,\z}\delta_{\e,\z},\label{eq:matrixelements}
    \end{align}
    where we used the properties of $X$ and $Z$ as boost and shift operators to go from the first to the second line, the fact that $\zxbraket{\g}{\c} = (-1)^{\g\cdot\c}/2^{n/2}$ to go from the second to the third line, and the fact that $\sum_{\g}(-1)^{\g\cdot\c} = 2^{n}\delta_{\c,\z}$ to go from the third to the fourth line. Eq. \eqref{eq:matrixelements} establishes Eq. \eqref{eq:twooperatorsneedingtobeequal}, completing the proof.
\end{proof}

We now use this lemma to prove Theorem 1.1.

\begin{Lemma}[DGBR Connection]
    Let $W_{\g,\c}$ and $\KD_{\g,\c}$ denote the DGBR distribution and the $\G$-KD distribution respectively. Then  $\Re\left[\KD_{\g,\c}(\rho)\right] = W_{\g,\c}(\rho)$.
\end{Lemma}

\begin{proof}


    Substituting Eq. \eqref{eq:hermitianization_equals} into Eq. (2) of the main article and evaluating the trace in the computational basis yields
    \begin{align}
        W_{\g,\c}(\rho) &= \Tr\left[\frac{1}{2}(B_{\g,\c} + B_{\g,\c}^{\dagger})\rho\right]\nonumber\\
        &= \frac{1}{2} \sum_{\h,\h'} \zbra{\h}B_{\g,\c}\zket{\h'} \Re \left(\zbra{\h'}\rho\zket{\h}\right) + \frac{1}{2}i \sum_{\h,\h'} \zbra{\h}B_{\g,\c}\zket{\h'} \Im \left(\zbra{\h'}\rho\zket{\h}\right)\nonumber\\
        &+ \frac{1}{2} \sum_{\h,\h'} \zbra{\h'}B_{\g,\c}\zket{\h} \Re \left(\zbra{\h'}\rho\zket{\h}\right) + \frac{1}{2} i\sum_{\h,\h'} \zbra{\h'}B_{\g,\c}\zket{\h} \Im \left(\zbra{\h'}\rho\zket{\h}\right)
    \end{align}
    where we used that the matrix elements of the phase-space point operators of the $\G$-KD distribution are real in the computational basis. Reindexing the third and fourth sum and using that the real and imaginary parts of $\rho$ are symmetric and anti-symmetric respectively, we find that
    \begin{align}
        W_{\g,\c}(\rho)  &= \frac{1}{2} \sum_{\h,\h'} \zbra{\h}B_{\g,\c}\zket{\h'} \Re \left(\zbra{\h'}\rho\zket{\h}\right) + \frac{1}{2}i \sum_{\h,\h'} \zbra{\h}B_{\g,\c}\zket{\h'} \Im \left(\zbra{\h'}\rho\zket{\h}\right)\nonumber\\
        &+ \frac{1}{2} \sum_{\h,\h'} \zbra{\h}B_{\g,\c}\zket{\h'} \Re \left(\zbra{\h'}\rho\zket{\h}\right) - \frac{1}{2}i \sum_{\h,\h'} \zbra{\h}B_{\g,\c}\zket{\h'} \Im \left(\zbra{\h'}\rho\zket{\h}\right)\nonumber\\
        &= \Re\sum_{\h,\h'} \zbra{\h}B_{\g,\c}\zket{\h'} \Re \left(\zbra{\h'}\rho\zket{\h}\right)\nonumber\\
        &= \Re \Tr(B_{\g,\c}\rho)\nonumber\\
        &= \Re\left[ \KD_{\g,\c}(\rho)\right]
    \end{align}
\end{proof}

\section{Proof of Hudson's Theorem for the $\G$-KD Distribution}
We prove Theorem 1.2 of the main article in this Supplemental Note.
\begin{Lemma}[Hudson's Theorem for the $\G$-KD Distribution] 
    Let $\ket{\psi} \in \HilbertSpace$. $\ket{\psi}$ is KD positive if and only if $\ket{\psi}$ is a CSS state.
\end{Lemma}

\begin{proof}
    In \cite{debièvre2025}, it was shown that for a KD distribution based on the group $G$, $\ket{\psi}$ is KD positive if and only if it is of the form 
    \begin{equation}\label{eq:KD_positive_pure_states_supp}
        \ket{H; \g, \c}\coloneqq \frac{1}{|H|^{1/2}}\Pauli_{\g,\c}\sum_{\h\in H}\zket{\h} \coloneqq\Pauli_{\g,\c}\ket{H},
    \end{equation}
    where $H$ is a subgroup of $G$. We are thus left to show that (i) $\ket{H; \g, \c}$ is a CSS state, and that (ii) all CSS states are of the form $\ket{H; \g, \c}$.

    We start with (i). To determine if $\ket{H; \g, \c}$ is a CSS state, we need to find its maximal stabilizing subgroup. We first focus on $\ket{H}$. Let $H^{\perp}$ denote the set of characters in $\hat{G}$ that are trivial on $H$. That is $H^{\perp} \coloneqq \{\c \in \hat{G}| \c(\h) = 1\forall \h \in H \}$. We consider the following subgroup of the Pauli group:
    \begin{equation}\label{eq:stabilizing_subgroup_of_H}
        \langle\{\Pauli_{\h,\z}\}_{\h\in H} \cup\{\Pauli_{\z,\e}\}_{\e\in H^{\perp}} \rangle  = \{P_{\h,\e}| h\in H, \e \in H^{\perp}\}.
    \end{equation}

    Since $H$ is closed under group addition and since group addition is an isomorphism, we find that $\Pauli_{\h,\z}\ket{H} = \ket{H+\h} = \ket{H}$. Furthermore, since $\e \in H^{\perp}$, we also find that $\Pauli_{\z,\e}\ket{H} = \ket{\e(H)} = \ket{H}$. Hence Eq. \eqref{eq:stabilizing_subgroup_of_H} is a stabilizing subgroup. To show that this stabilizing subgroup indeed uniquely determines $\ket{H}$, we must show that it is maximal in the sense that its size is $2^n$ \cite{Aaronson2004}. 

    In view of Eq. \eqref{eq:stabilizing_subgroup_of_H}, there are $|H| \times |H^{\perp}|$ distinct stabilizers. We must therefore show that $|H|\times  |H^{\perp}| = 2^n$.

    To this end, we consider the group homomorphism $\phi: \hat{G}\rightarrow\hat{H}$ defined by the restriction map $\c \mapsto \c|_H$, where $\hat{H}$ is the character group of $H$. Clearly, $\ker \phi = H^{\perp}$ and $\range \phi = \hat{H}$, where $\range$ denotes the range (image) of a function. The first isomorphism theorem then implies that
    \begin{equation}
        \hat{G}/H^{\perp} = \hat{G}/\ker \phi \cong \range \phi = \hat{H} \cong H
    \end{equation}
    from which it follows that $|H|\times|H^{\perp}| = |G| = 2^n$. The stabilizing subgroup Eq. \eqref{eq:stabilizing_subgroup_of_H} is therefore maximal and thus uniquely defines $\ket{H}$.

    To characterize the stabilizing subgroup of $\ket{H; \g, \c}$, note that if $S$ stabilizes $\ket{\psi}$, then $\Pauli_{\g,\c}S\Pauli_{\g,\c}^\dagger$ stabilizes $\Pauli_{\g,\c} \ket{\psi}$. The maximal stabilizing subgroup for $\ket{H;\g,\c}$ is therefore
    \begin{equation}\label{eq:Stabilizing_subgroup_of_Hgc}
        \langle\{(-1)^{\h\cdot\c}\Pauli_{\h,\z}\}_{\h\in H} \cup\{(-1)^{\g\cdot\e}\Pauli_{\z,\e}\}_{\e\in H^{\perp}} \rangle,
    \end{equation}
    where we used that $\Pauli_{\g,\c}\Pauli_{\h,\e} = (-1)^{\g\cdot\e + \h\cdot\c}\Pauli_{\h,\e}\Pauli_{\g,\c}$ and that $\Pauli_{\g,\c}\Pauli_{\g,\c}^\dagger = I$. The stabilizing subgroup of $\ket{H; \g, \c}$ therefore is of the form $\langle U \cup V\rangle$, where $U$ and $V$ are sets containing Pauli operators with only $Z$s and $X$s respectively. Hence $\ket{H; \g, \c}$ is a CSS state.

    We move on to proving the converse statement (ii). If $\ket{\psi}$ is a CSS state, then it has a stabilizing subgroup of the form $\langle S_Z \cup S_X\rangle$, where we may take $S_Z$ and $S_X$ to be subgroups of the Pauli group of the form $S_Z = \{(-1)^{f(\h)}\Pauli_{\h,\z}\}_{\h\in H}$ and $S_X = \{(-1)^{g(\e)}\Pauli_{\z,\e}\}_{\e\in K}$, where $f$ and $g$ are binary functions and $H$ and $K$ are subsets of $G,\hat{G}$ respectively \cite{Calderbank98, Steane1996, gottesman1997}. To show that this maximal stabilizing subgroup corresponds to a state of the form $\ket{H;\g,\c}$, we must show that it can be written in the form \eqref{eq:Stabilizing_subgroup_of_Hgc}. That is, we need to show that $H$ is a subgroup of $G$; that $K = H^{\perp}$; and that there exist $\g, \c \in G, \hat{G}$ such that $f(\h) = \c\cdot\h$ and $g(\e) = \g \cdot \e$.

    We start by showing that the subset $H$ is a subgroup of $G$. To do this, we must show that $H$ is closed under group addition and inversion and that $H$ must contain the identity. Suppose $\h,\h' \in H$, then
    \begin{equation}\label{eq:PaulihonPaulih}
        (-1)^{f(\h)}\Pauli_{\h,\z}(-1)^{f(\h')}\Pauli_{\h',\z} = (-1)^{f(\h)+f(\h')}\Pauli_{\h + \h',\z} \in S_Z.
    \end{equation}
    Hence, $\h + \h' \in H$, and closure under group addition follows. Furthermore, note that $\forall \h \in H, -\h = \h$ since $H$ is a subset of $G = \G$. Hence, $H$ is also closed under inversion. Last, since $H$ is closed under group addition and $\h + \h = \z$, we have that $H$ also contains the identity. Hence, $H$ is a subgroup of $G$. Similar arguments show that $K$ is a subgroup of $\hat{G}$.

    Next, we prove that $K = H^{\perp}$. Let $\bm{k} \in K$. Note that since $\langle S_Z \cup S_X\rangle$ is a stabilizing subgroup, all its elements must commute. Hence, $\forall \h \in H$, $\Pauli_{\bm{k},\z}\Pauli_{\h,\z} = \Pauli_{\h,\z}\Pauli_{\bm{k},\z}$. But from the Pauli commutation relationships, we also know that $\Pauli_{\bm{k},\z}\Pauli_{\h,\z} = (-1)^{\bm{k}\cdot\h}\Pauli_{\h,\z}\Pauli_{\bm{k},\z}$. This implies that $\forall \h \in H$, $\bm{k}(\h) = (-1)^{\bm{k}\cdot\h} = 1$. Hence, $\bm{k} \in H^{\perp}$ and thus $K \subset H^{\perp}$. Furthermore, since $\langle S_Z \cup S_X\rangle$ is a maximal stabilizing subgroup, we must have that $|K| = |H^{\perp}|$. It follows that $K = H^{\perp}$.

    Last, we must show that there exist $\g, \c \in G, \hat{G}$ such that $f(\h) = \c\cdot\h$ and $g(\e) = \g \cdot \e$. From Eq. \eqref{eq:PaulihonPaulih}, it follows that $f(\h) + f(\h') = f(\h + \h') $, i.e. $f$ is a linear map from $H$ to the set $\{0,1\}$. It follows that there must exist a $\c \in \hat{G}$ such that $f(\h) = \c\cdot\h$. The same argument also applies to $g$. We have thus shown that all CSS states are of the form $\ket{H;\g,\c}$, completing the proof.
\end{proof}
\section{Proof of the Covariance Relationships}
We prove Theorem 1.3 of the main article in this Supplemental Note.
\begin{Lemma}[Covariance]
    The $\G$-KD distribution obeys the following covariance relationships:
        \begin{align}\label{eq:covariance_of_paulis_supp}
            \KD_{\g_0,\c_0}(\Pauli_{\g,\c}\rho \Pauli_{\g,\c}^{\dagger}) &= \KD_{\g_0 + \g,\c_0 + \c}(\rho)\\ \label{eq:covariance_of_Hadamards_supp}
            \KD_{\g,\c}(\Hadamard^{\otimes n}\rho \Hadamard^{\otimes n\dagger}) &= \overline{\KD_{\c,\g}(\rho)} \\\label{eq:covariance_of_CNOTs_supp}
            \KD_{\g,\c}(\CX_{ct}\rho \CX_{ct}^\dagger) &= \KD_{A_{ct}\g, B_{ct}\c} (\rho),
        \end{align}
        where $(A_{ct})_{jk} \coloneqq \delta_{jk} + \delta_{tj}\delta_{ck}$ and $(B_{ct})_{jk} \coloneqq \delta_{jk} + \delta_{cj}\delta_{tk}$.
\end{Lemma}

\begin{proof}
    Covariance relation \eqref{eq:covariance_of_paulis_supp} follows directly from Eq. (29) in \cite{debièvre2025}. We move on to showing covariance relation $\eqref{eq:covariance_of_Hadamards_supp}$. The effect of the simultaneous Hadamard on such states is $\Hadamard^{\otimes n}\xket{\mathbf{a}} = \zket{\mathbf{a}}$. We thus find that
    \begin{align}
        \KD_{\g,\c}(\Hadamard^{\otimes n}\rho \Hadamard^{\otimes n\dagger}) &= \xbra{\c}\Hadamard^{\otimes n}\Hadamard^{\otimes n}\zket{\g}\zbra{\g}\Hadamard^{\otimes n}\rho\Hadamard^{\otimes n}\xket{\c\nonumber}\\
        &=\zxbraket{\c}{\g}\xbra{\g}\rho\zket{\c}\nonumber\\
        &= \overline{\xzbraket{\g}{\c}}\overline{\zbra{\c}\rho\xket{\g}}\nonumber\\
        &= \overline{\KD_{\c,\g}(\rho)},
    \end{align}
    proving the result.
    Finally we prove covariance relation \eqref{eq:covariance_of_CNOTs_supp}. To do this, we use that $\CX \zket{a}\otimes\zket{b} =\zket{a}\otimes\zket{(a + b)}$ and that $\CX \xket{a}\otimes\xket{b} =\xket{(a + b)}\otimes\xket{b}$. This implies that $\CX_{ct}\zket{\g} = \zket{A_{ct}\g}$ and $\CX_{ct}\xket{\c} = \xket{B_{ct}\c}$, where $(A_{ct})_{jk} \coloneqq \delta_{jk} + \delta_{tj}\delta_{ck}$ and $(B_{ct})_{jk} \coloneqq \delta_{jk} + \delta_{cj}\delta_{tk}$. We thus find that
    \begin{align}
        \KD_{\g,\c}(\CX_{ct}\rho \CX_{ct}^\dagger) &=\xbra{\c}\CX_{ct}\CX_{ct}\zket{\g}\zbra{\g}\CX_{ct}\rho\CX_{ct}\xket{\c}\nonumber\\
        &=\xzbraket{B_{ct}\c}{A_{ct}\g}\zbra{A_{ct}\g}\rho\xket{B_{ct}\c}\nonumber\\
        &= \KD_{A_{ct}\g, B_{ct}\c}
    \end{align}
    establishing the last covariance relation.
\end{proof}

\section{KD-Based Simulation Algorithm}
It was shown in the main article that KD-positive input states may be efficiently simulated using the simulation algorithm built on the DGBR distribution from Ref. \cite{Delfosse2015}. We now show that the DGBR simulation algorithm may be phrased entirely in terms of the $\G$-KD distribution.

\begin{algorithm}[H]
    \caption{Phase-Space Simulation Algorithm}\label{alg:phase_space_simulator}
    \begin{algorithmic}[1]
    \Statex \textbf{Input}: DGBR circuit $C$ and KD-positive input state $\inputstate$
    \Statex \textbf{Output}: bitstring $(b_j)_j$
    \State Sample $\g, \c$ from $Q(\inputstate)$ 
    \For{each gate in circuit $C$}
        \State \textbf{If} the gate is $\Pauli_{\g',\c'}$, set $\g \leftarrow \g + \g'$ and $\c \leftarrow \c + \c'$ 
        \State \textbf{If} the gate is $\Hadamard$, set $\g,\c \leftarrow \c,\g$
        \State \textbf{If} the gate is $\CX_{ct}$, set $\g \leftarrow A_{ct}\g$ and $\c \leftarrow B_{ct} \c$
        \State \textbf{If} the gate is a $Z$-measurement of the $j$th qubit:
        \State \indent Output the bit $b = g_j$
        \State \indent With probability $1/2$, set $\c \leftarrow \c + \mathbf{e}_j$.
    \EndFor
    \end{algorithmic}
\end{algorithm}

In the above, $\mathbf{e}_j$ is the vector with a $1$ in the $j$th entry and zero elsewhere, and $A_{ct}$ and $B_{ct}$ are defined in the main article. This algorithm runs in linear time. To prove its correctness, we need the following lemma:

\begin{Lemma}\label{lem:updating_KD_distributions} If the $j$th qubit of a state $\rho$ is measured in the computational basis and we find output $b = g_j$, then the $\G$-KD distribution of the updated state $\rho'$  is 
    \begin{equation}
        \KD_{\g,\c}(\rho') = \frac{1}{2}\left[\KD_{\g,\c}(\rho) + \KD_{\g,\c+\mathbf{e}_j}(\rho)\right].
    \end{equation}
\end{Lemma}

\begin{proof}
    After the measurement, the state is updated to
    \begin{align}
        \rho' &\propto \left[I+(-1)^{b}Z_j\right]\rho\left[I+(-1)^{b}Z_j\right] \nonumber\\
              &= \rho + (-1)^{b}Z_j\rho + (-1)^{b}\rho Z_j + Z_j \rho Z_j.\label{eq:updated_state}
    \end{align}
    One may straightforwardly calculate that
    \begin{equation}\label{eq:KD_Zrho}
        \KD_{\g,\c}(Z_j\rho) = (-1)^{g_j} \KD_{\g,\c}(\rho),
    \end{equation}
    and
    \begin{equation}\label{eq:KD_rhoZ}
        \KD_{\g,\c}(\rho Z_j) = (-1)^{g_j} \KD_{\g,\c +\mathbf{e}_j}(\rho).
    \end{equation}
    
    Taking the KD-distribution of Eq. \eqref{eq:updated_state}, and substituting in Eqs. \eqref{eq:KD_Zrho} and \eqref{eq:KD_rhoZ} together with covariance relation \eqref{eq:covariance_of_paulis_supp} then yields:
    \begin{equation}
        \KD_{\g,\c}(\rho') \propto \KD_{\g,\c}(\rho) + \KD_{\g,\c+\mathbf{e}_j}(\rho).
    \end{equation}
    Normalization yields the desired result.
\end{proof}

We now prove the correctness of our simulation algorithm:

\begin{Theorem}\label{thm:correctness}
    Let $C$ be a DGBR circuit with input state $\rho$. If $\rho$ is KD positive and we can sample efficiently from $\KD(\rho)$, then there exists an efficient classical algorithm that samples from the same probability distribution as the quantum circuit.
\end{Theorem}

\begin{proof}
    We prove this by showing that Algorithm \ref{alg:phase_space_simulator} samples from the correct output distribution. We will use an approach similar to the one taken for the proof of Theorem 1 in Ref. \cite{Veitch2012}. Let $\rho_p$ denote the state of the qubits after the $p$th gate is applied. We first show by induction that after $p$ steps, $\g$ and $\c$ are distributed according to $\KD(\rho_p)$. This is trivially true for $p = 0$, establishing the base case. Assume it is true for $p$. If we apply a Pauli gate $\Pauli_{\g_0,\c_0}$ at time step $p$, then $\KD_{\g,\c}(\rho_{p+1}) = \KD_{\g\c}(\Pauli_{\g,\c}\rho_p \Pauli_{\g_0,\c_0}^{\dagger}) = \KD_{\g+\g_0,\c+\c_0}(\rho_p)$, where we used covariance relationship \eqref{eq:covariance_of_paulis_supp} in the second equality. It follows that updating $\g$ to $\g + \g_0$ and $\c$ to $\c + \c_0$ ensures that $\g$ and $\c$ are distributed according to $\KD_{\g,\c}(\rho_{p+1})$. Similar arguments apply to when $\Hadamard^{\otimes n}$ or $\CX_{ct}$ are applied.

    If a measurement is applied at timestep $t$, then we obtain the outcome $b = g$. Using Lemma \ref{lem:updating_KD_distributions}, we know that $\KD_{\g,\c}(\rho_{p+1})$ is distributed according to $\frac{1}{2}\left[\KD_{\g,\c}(\rho_p) + \KD_{\g,\c+\mathbf{e}_j}(\rho_p)\right]$. Thus updating $\c$ to $\c + \mathbf{e}_j$ with a probability 1/2 indeed ensures that $\g$ and $\c$ are distributed according to $\KD_{\g,\c}(\rho_{p+1})$. This completes the proof by induction.

    We are left to show that given a state $\rho_p$, measuring the $j$th qubit in the computational basis is indeed the same as sampling $\g,\c$ from $\KD(\rho)$ and returning $g_j$. The probability of getting outcome $b$ when measuring qubit $j$ is
    \begin{align}
        \mathbb{P}(Z_j = b | \rho_p)&= \Tr \left[\frac{1}{2}(1+ (-1)^{b}Z_j)\rho_p\right]\nonumber\\
                                    &= \sum_{\g,\c} \tilde{\KD}_{\g,\c}\left(\frac{1}{2}(1+ (-1)^{b}Z_j)\right) \KD_{\g,\c}(\rho_p),
    \end{align}
    where the overlap formula (Eq. (8) of the main article) was used. A straightforward calculation shows that
    \begin{align}
        \tilde{\KD}_{\g,\c}\left(\frac{1}{2}(1+ (-1)^{b}Z_j)\right) &= \frac{\xbra{\chi_j}I+(-1)^{b}Z\zket{g_j}}{\xzbraket{\chi_j}{g_j}}\nonumber\\
        &= \delta_{g_j,b}.
    \end{align}
    We thus find that 
    \begin{equation}
        \mathbb{P}(Z_j = b | \rho_p) = \sum_{\g,\c: g_j = b} \KD_{\g\c}(\rho_p).
    \end{equation}
    \textit{Id est}, the probability of obtaining outcome $b$ is equal to sampling outcome $\g,\c$ from $\KD(\rho)$ and obtaining $g_j = b$.
\end{proof}

\section{Analytical Construction of Bound States}
In the main article, we considered an operator $F$ normal to a facet of both the CSS and rebit stabilizer polytopes:
\begin{equation}
    F = 
    \begin{pmatrix}
        1         &  0          &  1            &  1         \\
        0         &  1          & -1            & -1         \\
        1         & -1          & -1            & -2         \\
        1         & -1          & -2            & -1
\end{pmatrix}.
\end{equation}
The half-space corresponding to this facet is defined by the inequality $\Tr(\rho F) \leq 1$.
For all $24$ rebit stabilizer states on $2$ qubits and for all $20$ CSS states on $2$ qubits, $\Tr(\rho F) \leq 1$. The inequality is tight for at least 9 states in each case. Since the dimensionality of the affine space the rebit states live in is $9$, this means that $\Tr(\rho F) \leq 1$ indeed defines a facet of both polytopes.

We considered the state $\rho_\lambda$, given by
\begin{equation}
    \rho_{\lambda} = \frac{1}{4}I + \lambda F,
\end{equation}
and noted that this state lies outside the rebit stabilizer polytope for $\lambda > 1/20$. Since the rebit stabilizer polytope contains all real stabilizer mixtures, and since $\rho_\lambda$ is real for all $\lambda$, this implies that $\rho_\lambda$ is magic for $\lambda > 1/20$.

Finally, to show that $\rho_\lambda$ is a bound magic state, we must find the $\lambda$ for which $\rho_\lambda$ is a KD-positive state. We start with KD-positivity. The KD-distribution of $\rho_\lambda$ is given by
\begin{equation}
    \KD(\rho_\lambda) = \frac{1}{15}
    \begin{pmatrix}
        1+12\lambda  &  1+ 4\lambda  &  1- 4\lambda  &  1+ 4\lambda  \\
        1- 4\lambda  &  1+ 4\lambda  &  1+12\lambda  &  1+ 4\lambda  \\
        1-12\lambda  &  1+12\lambda  &  1-12\lambda  &  1- 4\lambda  \\
        1-12\lambda  &  1- 4\lambda  &  1-12\lambda  &  1+12\lambda
    \end{pmatrix}.
\end{equation}
$\rho_\lambda$ therefore is a KD positive operator for $\lambda \in [-1/12,1/12]$. 

Next, we need to find the values of $\lambda$ for which $\rho_\lambda$ is indeed a quantum state. By construction, $\rho_\lambda$ is Hermitian and trace $1$ for all $\lambda$, so for $\rho_\lambda$ to be a quantum state, we only need to ensure that it is positive semidefinite. We may ensure positive semidefiniteness by requiring all the eigenvalues of $\rho_\lambda$ to be nonnegative. Letting $\mu_j$ denote the eigenvalues of $\rho_\lambda$, the characteristic equation shows that 
\begin{equation}
    \det(\lambda F - (\mu_j -1/4)I) = 0.
\end{equation}
Denoting the eigenvalues of $F$ by $\nu_j$, we thus find that
\begin{equation}\label{eq:eigenvaluesofrhoandF}
    \mu_j = \lambda \nu_j + 1/4. 
\end{equation}
The minimal eigenvalue of $F$ may be calculated to be $-1-2\sqrt{2}$ using standard methods. Using Eq. \eqref{eq:eigenvaluesofrhoandF}, we find that $\rho_\lambda$ is positive semidefnite and thus a state for $\lambda \in [0,1/(4+8\sqrt{2})]$. This is a stronger requirement than KD-positivity, so we find that $\rho_\lambda$ is a KD positive state for $\lambda \in [0,1/(4+8\sqrt{2})]$.

Putting these results together, we find that $\rho_\lambda$ is both magic and a KD-positive state for $\lambda \in ]1/20,1/(4+8\sqrt2)]$. Since KD-positive states are simulable, $\rho_\lambda$ must therefore be a bound magic state for $\lambda \in ]1/20,1/(4+8\sqrt2)]$.

\section{Numerical Construction of Bound States above CSS Facets}
In this Supplemental Note, we numerically show that there exist bound states above every facet shared between the CSS and rebit stabilizer polytope. To do this, we first need an explicit description of the rebit stabilizer and CSS polytopes. We may express any $2$-qubit state $\rho$ in the Pauli basis using
\begin{equation}
    \rho = \sum_{P_1,P_2 \in \{I,X,Y,Z\}} r_{P_1,P_2}(\rho)P_1\otimes P_2,
\end{equation}
where $r_{P_1,P_2}(\rho)$ are the coordinates of $\rho$ in this basis. Note that $r_{P_1,P_2}(\rho)$ must be real and that $r_{I,I}(\rho) = 1/4$. This means that the coordinates of the stabilizer states live in an affine subspace of $\mathbb{R}^{16}$. 

To find all the facets of the rebit stabilizer and CSS polytopes, we express the rebit stabilizer and CSS states in the Pauli basis. There are $24$ 2-qubit rebit stabilizer states, of which $20$ are CSS. We thus get $24$ points in $\mathbb{R}^{16}$, of which $20$ belong to CSS states. We use cdd \cite{Fukuda1995}, a numerical package for analyzing polytopes in $\mathbb{R}^m$, to find the facets of these polytopes. We find that the rebit stabilizer polytope has $120$ facets and the CSS polytope has $40$ facets. Between these two sets of facets, exactly $24$ facets coincide.

We take the normal of each shared facet $F$ and consider a state $\rho_\lambda = I/4 + \lambda F$ along this normal as in Note~V. Using a bisection method, we numerically find $\lambda_{\textrm{KD+}}$, the maximal $\lambda$ for which $\rho_\lambda$ is KD positive. Similarly, we also numerically find $\lambda_{\textrm{SD}}$, the largest value of $\lambda$ for which $\rho_\lambda$ only has positive eigenvalues and thus is a valid density matrix.

To determine whether a state is magic, we consider the following linear program:
\begin{align}\label{eq:linear_program_for_deciding_magic}
    \max_{p_S} \;0 &\;\;\textrm{  s.t.  }\;\;\sum_S p_S S = \rho_\lambda,\nonumber\\ 
                   &\;\;\textrm{  and  }\;\; \sum_S p_S = 1,\nonumber\\
                   &\;\;\textrm{  and  }\;\; p_S \geq 0,
\end{align}
where the sums are taken over all $2$-qubit stabilizer states. If this linear program is feasible, then $\rho_\lambda$ must be a stabilizer mixture. If the linear program is infeasible, then $\rho_\lambda$ must be magic. Using a bisection method with a linear program solver as a subroutine, we numerically find $\lambda_{\textrm{magic}}$, the largest value of $\lambda$ for which $\rho_\lambda$ is a stabilizer mixture.

For every shared facet, we find $\lambda_{\textrm{KD+}} = 0.083$, $\lambda_{\textrm{SD}} = 0.065$ and $\lambda_{\textrm{magic}} = 0.050$. These values agree up to numerical precision with those obtained using the analytical approach in Note V. Since $\lambda_{\textrm{magic}} \leq \lambda_{\textrm{KD+}}$ and $\lambda_{\textrm{magic}} \leq \lambda_{\textrm{SD}}$, we have found bound states above every shared facet of the stabilizer polytope.

\section{Estimation of Relative Volume of Bound States}
In this Supplemental Note, we describe our methodology for obtaining the data presented in Table I of the main article. We aim to estimate the relative volumes that the four simulation categories occupy in the two-qubit rebit state space. Here, the volume is estimated by the ratio of the number of random samples falling into a given category to the total number of sampled states. 


To ensure unbiased sampling, we generate random rebit density matrices using the real Ginibre ensemble. This construction guarantees uniform sampling under the Hilbert-Schmidt measure, as formally established in Ref.~\cite{zyczkowski2001induced}. Let $A \in \mathbb{R}^{N \times N}$ be a random matrix with entries independently drawn from the standard normal distribution $\mathcal{N}(0, 1)$. The real Ginibre ensemble
\begin{equation}
\rho = \frac{A A^\top}{\mathrm{Tr}[A A^\top]}
\end{equation}
is a valid rebit density matrix: real, symmetric, positive semidefinite, and normalized. The following algorithm describes the practical generation of such rebit states:

\begin{algorithm}[H]
\caption{Sampling a Random $n$-Rebit Density Matrix via Ginibre Ensemble}
\label{sample}
\label{sample_ginibre}
\begin{algorithmic}[1]
\Require Number of rebits $n$
\Ensure A valid $n$-rebit density matrix $\rho \in \mathbb{R}^{2^n \times 2^n}$
\State Let $d \gets 2^n$
\State Sample $A \in \mathbb{R}^{d \times d}$ where each $A_{ij} \sim \mathcal{N}(0, 1)$ \Comment{Real Ginibre matrix}
\State Compute: $\rho \gets A A^\top$ \Comment{Ensure symmetry and positivity}
\State Normalize: $\rho \gets \rho / \operatorname{Tr}[\rho]$ \Comment{Ensure unit trace}
\State \Return $\rho$
\end{algorithmic}
\end{algorithm}

The computational bottleneck in our simulations was the classification of the random samples. After obtaining a randomly sampled rebit state, we first compute the KD distribution to verify its positivity. Then, we apply the methodology described in Eqn. \eqref{eq:linear_program_for_deciding_magic} to determine whether the state lies within the stabilizer polytope. Sampling and state classifying 1 billion $2$-qubit states took approximately 3.5 days on a x86\_64 processor with 24 physical cores and 32 logical CPUs on Intel(R) Core(TM) i9-13900K.

\section{Proof of Theorem 2}
In this Supplemental Note, we prove Theorem 2 of the main article. We need to show that KD mana is faithful, additive, and nonincreasing under free operations. We start by showing that it is faithful.

\begin{Lemma}
    Let $\rho$ be a quantum state. Then $\mana(\rho) = 0 $ if and only if $\rho$ is KD positive.
\end{Lemma}

\begin{proof}
    We prove the forward direction by proving the contrapositive: if $\rho$ is KD nonpositive, then $\mana(\rho) \neq 0$. Suppose $\rho$ is KD nonpositive. Then there must exist a $\delta> 0$ such that there exist $\g_0,\c_0 \in G\times\hat{G}$ such that either $|\Re \left[\KD_{\g_0,\c_0}(\rho)\right]| - \Re \left[\KD_{\g_0,\c_0}(\rho)\right] \geq \delta$ or $|\Im\left[\KD_{\g_0,\c_0}(\rho)\right] |\geq \delta$. This implies that there exist $\delta, \g_0, \c_0$ such that $|\KD_{\g,_0\c_0}(\rho)| - \Re\left[\KD_{\g_0,\c_0}(\rho)\right] \geq \delta$.

    Next, we consider minus one plus the sum of the absolute values of the entries of the KD-distribution:
    \begin{equation}\label{eq:KDnegexpanded}
        -1 + \sum_{\g,\c}|\KD_{\g,\c}(\rho)|  = \sum_{\g,\c}\left(|\KD_{\g,\c}(\rho)| - \Re\left[\KD_{\g,\c}(\rho)\right]\right),
    \end{equation}
    where we used that the real part of the KD distribution sums to 1. Each term in the sum on the right-hand side of Eq. \eqref{eq:KDnegexpanded} is nonnegative. Since $\rho$ is KD nonpositive, at least one of the terms in the sum must be greater than $\delta$. We thus find that $-1 + \sum_{\g,\c}|\KD_{\g,\c}(\rho)| \geq \delta$. Rearranging and taking logarithms yields
    \begin{equation}
        \mana(\rho) \geq \log (1+\delta) \neq 0,
    \end{equation}
    since $\delta >0$. We have thus proven the contrapositive, establishing the result.

    We now move on to the backward direction. If $\rho$ is KD positive, then $\KD_{\g,\c}(\rho) \in [0,1]$. Hence $|\KD_{\g,\c}(\rho)| = \KD_{\g,\c}(\rho)$. Since the KD distribution is normalized to one, we find that $\sum_{\g,\c}|\KD_{\g,\c}(\rho)| = \sum_{\g,\c}\KD_{\g,\c}(\rho) =1$. The result follows upon taking logarithms of both sides.
\end{proof}

Next, we prove that the KD mana is additive.

\begin{Lemma}
    Let $\rho, \sigma$ be quantum states. Then, $\mana(\rho\otimes\sigma) = \mana(\rho) + \mana(\sigma)$.
\end{Lemma}

\begin{proof}
    Let $\g,\c$ denote the indices for the KD distribution of $\rho$, and let $\g',\c'$ denote the indices for the KD distribution of $\sigma$. Evaluating the KD mana of $\rho\otimes\sigma$ then yields
    \begin{align}
        \mana(\rho\otimes\sigma) &= \log \sum_{\g\g'\c\c'}|\KD_{(\g,\g'),(\c,\c')}(\rho\otimes\sigma)|\nonumber\\
                                 &= \log \left(\sum_{\g\c}|\KD_{\g,\c}(\rho)|\sum_{\g'\c'}|\KD_{\g',\c'}(\sigma)|\right)\nonumber\\
                                 &= \mana(\rho) + \mana(\sigma),
    \end{align}
    where we used Theorem 1.4 of the main paper to go from the first to the second line.
\end{proof}




Finally, we prove that the KD-mana is a resource monotone.

\begin{Lemma}
    $\mana$ is nonincreasing under the free operations and is thus a monotone.
\end{Lemma}

\begin{proof}
    We need to show that $\mana$ is nonincreasing under the action of the unitaries $\Pauli_{\g,\c}, \Hadamard^{\otimes n}$, and $\CX$, under measurements in the computational basis, under the initialization of CSS states and under partial traces.

    From Theorem 1.3, we know that the unitaries $\Pauli_{\g,\c}, \Hadamard^{\otimes n}, \CX$ permute the entries of the KD-distribution. The KD mana sums over all entries of the KD distribution, and is therefore insensitive to the ordering of the entries. Hence $\mana$ is invariant under the action of these unitaries.

    We now show that $\mana$ is nonincreasing under the action of computational basis measurements. Lemma \ref{lem:updating_KD_distributions} shows that if the $j$th qubit of a state $\rho$ is measured, then the KD distribution of the updated state $\rho'$ is
    \begin{equation}
        \KD_{\g,\c}(\rho') = \frac{1}{2}\left[\KD_{\g,\c}(\rho) + \KD_{\g,\c+\mathbf{e}_j}(\rho)\right],
    \end{equation}
    where $\mathbf{e}_j$ is the vector containing a $1$ in the $j$th entry and zero everywhere else. Evaluating the sum of the absolute values of the entries of the KD distribution then yields
    \begin{align}
        \sum_{\g,\c}|\KD_{\g,\c}(\rho')| &= \frac{1}{2} \sum_{\g,\c}|\KD_{\g,\c}(\rho) + \KD_{\g,\c+\mathbf{e}_j}(\rho)|\nonumber\\
        &\leq \frac{1}{2}\sum_{\g,\c}|\KD_{\g,\c}(\rho)|  + \frac{1}{2}\sum_{\g,\c}|\KD_{\g,\c+\mathbf{e}_j}(\rho)|\nonumber\\
        &= \sum_{\g,\c}|\KD_{\g,\c}(\rho)|.
    \end{align}
    Taking logarithms of both sides then yields the desired result.

    To show that $\mana$ is nonincreasing when CSS states are appended, note that if $\sigma$ is a CSS state, then $\mana(\sigma) = 0$ by Theorem 1.2 and Theorem 2.1 of the main article. Additivity then implies that for any $\rho$, $\mana(\rho\otimes\sigma) = \mana(\rho)$, showing that $\mana$ is constant and thus nonincreasing under the appending of CSS states.

    Last, we need to show that $\mana$ is nonincreasing under partial traces. Without loss of generality, we will consider tracing over the first qubit of the system. We let $g',\chi'$ denote the group element and character indexing the first qubit, and we let $\g,\c$ denote the group element and character indexing the other qubits. We may calculate KD distribution of the partial trace of a density matrix $\rho$ on $n$ qubits to be
    \begin{align}
        \KD_{\g,\c}\left(\sum_{g'}(\zbra{g'}\otimes I_{2^{n-1}}) \rho (\zket{g'}\otimes I_{2^{n-1}}) \right) &= \sum_{g'} \xzbraket{\g}{\c}(\zbra{g'}\otimes\xbra{\c}) \rho (\zket{g'}\otimes\zket{\g})\nonumber\\
        &= \sum_{g'\chi'}\zxbraket{g'\g}{\chi'\c} \xbra{\chi'\c}\rho\zket{g'\g}\nonumber\\
        &= \sum_{g'\chi'} \KD_{(g',\g),(\chi',\c)}(\rho).
    \end{align}
    Partial tracing over a qubit is therefore the same as summing over the relevant index in the KD distribution. This allows us to calculate the sum of the absolute values of the elements of a partially-traced KD distribution:
    \begin{align}
        \sum_{\g,\c}\left|\KD_{\g,\c}\left(\sum_{g'}(\zbra{g'}\otimes I_{2^{n-1}}) \rho (\zket{g'}\otimes I_{2^{n-1}}) \right)\right| &= \sum_{\g\c}\left|\sum_{g'\chi'} \KD_{(g',\g),(\chi',\c)}(\rho)\right|\nonumber\\
        &\leq \sum_{g'\g\chi'\c} \left| \KD_{(g',\g),(\chi',\c)}(\rho)\right|,
    \end{align}
    where we used the triangle inequality to go from the first line to the second line. The desired result follows upon taking logarithms, completing the proof.
\end{proof}


\begin{thebibliography}{73}%
\makeatletter
\providecommand \@ifxundefined [1]{%
 \@ifx{#1\undefined}
}%
\providecommand \@ifnum [1]{%
 \ifnum #1\expandafter \@firstoftwo
 \else \expandafter \@secondoftwo
 \fi
}%
\providecommand \@ifx [1]{%
 \ifx #1\expandafter \@firstoftwo
 \else \expandafter \@secondoftwo
 \fi
}%
\providecommand \natexlab [1]{#1}%
\providecommand \enquote  [1]{``#1''}%
\providecommand \bibnamefont  [1]{#1}%
\providecommand \bibfnamefont [1]{#1}%
\providecommand \citenamefont [1]{#1}%
\providecommand \href@noop [0]{\@secondoftwo}%
\providecommand \href [0]{\begingroup \@sanitize@url \@href}%
\providecommand \@href[1]{\@@startlink{#1}\@@href}%
\providecommand \@@href[1]{\endgroup#1\@@endlink}%
\providecommand \@sanitize@url [0]{\catcode `\\12\catcode `\$12\catcode `\&12\catcode `\#12\catcode `\^12\catcode `\_12\catcode `\%12\relax}%
\providecommand \@@startlink[1]{}%
\providecommand \@@endlink[0]{}%
\providecommand \url  [0]{\begingroup\@sanitize@url \@url }%
\providecommand \@url [1]{\endgroup\@href {#1}{\urlprefix }}%
\providecommand \urlprefix  [0]{URL }%
\providecommand \Eprint [0]{\href }%
\providecommand \doibase [0]{https://doi.org/}%
\providecommand \selectlanguage [0]{\@gobble}%
\providecommand \bibinfo  [0]{\@secondoftwo}%
\providecommand \bibfield  [0]{\@secondoftwo}%
\providecommand \translation [1]{[#1]}%
\providecommand \BibitemOpen [0]{}%
\providecommand \bibitemStop [0]{}%
\providecommand \bibitemNoStop [0]{.\EOS\space}%
\providecommand \EOS [0]{\spacefactor3000\relax}%
\providecommand \BibitemShut  [1]{\csname bibitem#1\endcsname}%
\let\auto@bib@innerbib\@empty
\bibitem [{\citenamefont {Nielsen}\ and\ \citenamefont {Chuang}(2010)}]{Nielsen2010}%
  \BibitemOpen
  \bibfield  {author} {\bibinfo {author} {\bibfnamefont {M.~A.}\ \bibnamefont {Nielsen}}\ and\ \bibinfo {author} {\bibfnamefont {I.~L.}\ \bibnamefont {Chuang}},\ }\href@noop {} {\emph {\bibinfo {title} {Quantum Computation and Quantum Information: 10th Anniversary Edition}}}\ (\bibinfo  {publisher} {Cambridge University Press},\ \bibinfo {year} {2010})\BibitemShut {NoStop}%
\bibitem [{\citenamefont {Bravyi}\ and\ \citenamefont {Kitaev}(2005)}]{Bravyi2005}%
  \BibitemOpen
  \bibfield  {author} {\bibinfo {author} {\bibfnamefont {S.}~\bibnamefont {Bravyi}}\ and\ \bibinfo {author} {\bibfnamefont {A.}~\bibnamefont {Kitaev}},\ }\bibfield  {title} {\bibinfo {title} {Universal quantum computation with ideal clifford gates and noisy ancillas},\ }\bibfield  {journal} {\bibinfo  {journal} {Physical Review A}\ }\textbf {\bibinfo {volume} {71}},\ \href {https://doi.org/10.1103/physreva.71.022316} {10.1103/physreva.71.022316} (\bibinfo {year} {2005})\BibitemShut {NoStop}%
\bibitem [{\citenamefont {Gottesman}(1998)}]{gottesman1998}%
  \BibitemOpen
  \bibfield  {author} {\bibinfo {author} {\bibfnamefont {D.}~\bibnamefont {Gottesman}},\ }\bibfield  {title} {\bibinfo {title} {The heisenberg representation of quantum computers},\ }\href {https://arxiv.org/abs/quant-ph/9807006} {\bibfield  {journal} {\bibinfo  {journal} {arXiv preprint}\ } (\bibinfo {year} {1998})},\ \Eprint {https://arxiv.org/abs/quant-ph/9807006} {arXiv:quant-ph/9807006 [quant-ph]} \BibitemShut {NoStop}%
\bibitem [{\citenamefont {Aaronson}\ and\ \citenamefont {Gottesman}(2004)}]{Aaronson2004}%
  \BibitemOpen
  \bibfield  {author} {\bibinfo {author} {\bibfnamefont {S.}~\bibnamefont {Aaronson}}\ and\ \bibinfo {author} {\bibfnamefont {D.}~\bibnamefont {Gottesman}},\ }\bibfield  {title} {\bibinfo {title} {Improved simulation of stabilizer circuits},\ }\bibfield  {journal} {\bibinfo  {journal} {Physical Review A}\ }\textbf {\bibinfo {volume} {70}},\ \href {https://doi.org/10.1103/physreva.70.052328} {10.1103/physreva.70.052328} (\bibinfo {year} {2004})\BibitemShut {NoStop}%
\bibitem [{\citenamefont {Reichardt}(2005)}]{Reichardt2005}%
  \BibitemOpen
  \bibfield  {author} {\bibinfo {author} {\bibfnamefont {B.~W.}\ \bibnamefont {Reichardt}},\ }\bibfield  {title} {\bibinfo {title} {Quantum universality from magic states distillation applied to css codes},\ }\href {https://doi.org/10.1007/s11128-005-7654-8} {\bibfield  {journal} {\bibinfo  {journal} {Quantum Information Processing}\ }\textbf {\bibinfo {volume} {4}},\ \bibinfo {pages} {251–264} (\bibinfo {year} {2005})}\BibitemShut {NoStop}%
\bibitem [{\citenamefont {Campbell}\ and\ \citenamefont {Browne}(2010)}]{Campbell2010}%
  \BibitemOpen
  \bibfield  {author} {\bibinfo {author} {\bibfnamefont {E.~T.}\ \bibnamefont {Campbell}}\ and\ \bibinfo {author} {\bibfnamefont {D.~E.}\ \bibnamefont {Browne}},\ }\bibfield  {title} {\bibinfo {title} {Bound states for magic state distillation in fault-tolerant quantum computation},\ }\href {https://doi.org/10.1103/PhysRevLett.104.030503} {\bibfield  {journal} {\bibinfo  {journal} {Phys. Rev. Lett.}\ }\textbf {\bibinfo {volume} {104}},\ \bibinfo {pages} {030503} (\bibinfo {year} {2010})}\BibitemShut {NoStop}%
\bibitem [{\citenamefont {Veitch}\ \emph {et~al.}(2012)\citenamefont {Veitch}, \citenamefont {Ferrie}, \citenamefont {Gross},\ and\ \citenamefont {Emerson}}]{Veitch2012}%
  \BibitemOpen
  \bibfield  {author} {\bibinfo {author} {\bibfnamefont {V.}~\bibnamefont {Veitch}}, \bibinfo {author} {\bibfnamefont {C.}~\bibnamefont {Ferrie}}, \bibinfo {author} {\bibfnamefont {D.}~\bibnamefont {Gross}},\ and\ \bibinfo {author} {\bibfnamefont {J.}~\bibnamefont {Emerson}},\ }\bibfield  {title} {\bibinfo {title} {Negative quasi-probability as a resource for quantum computation},\ }\href {https://doi.org/10.1088/1367-2630/14/11/113011} {\bibfield  {journal} {\bibinfo  {journal} {New Journal of Physics}\ }\textbf {\bibinfo {volume} {14}},\ \bibinfo {pages} {113011} (\bibinfo {year} {2012})}\BibitemShut {NoStop}%
\bibitem [{\citenamefont {Delfosse}\ \emph {et~al.}(2015)\citenamefont {Delfosse}, \citenamefont {Allard~Guerin}, \citenamefont {Bian},\ and\ \citenamefont {Raussendorf}}]{Delfosse2015}%
  \BibitemOpen
  \bibfield  {author} {\bibinfo {author} {\bibfnamefont {N.}~\bibnamefont {Delfosse}}, \bibinfo {author} {\bibfnamefont {P.}~\bibnamefont {Allard~Guerin}}, \bibinfo {author} {\bibfnamefont {J.}~\bibnamefont {Bian}},\ and\ \bibinfo {author} {\bibfnamefont {R.}~\bibnamefont {Raussendorf}},\ }\bibfield  {title} {\bibinfo {title} {Wigner function negativity and contextuality in quantum computation on rebits},\ }\href {https://doi.org/10.1103/PhysRevX.5.021003} {\bibfield  {journal} {\bibinfo  {journal} {Phys. Rev. X}\ }\textbf {\bibinfo {volume} {5}},\ \bibinfo {pages} {021003} (\bibinfo {year} {2015})}\BibitemShut {NoStop}%
\bibitem [{\citenamefont {Veitch}\ \emph {et~al.}(2014)\citenamefont {Veitch}, \citenamefont {Hamed~Mousavian}, \citenamefont {Gottesman},\ and\ \citenamefont {Emerson}}]{Veitch2014resourcetheory}%
  \BibitemOpen
  \bibfield  {author} {\bibinfo {author} {\bibfnamefont {V.}~\bibnamefont {Veitch}}, \bibinfo {author} {\bibfnamefont {S.~A.}\ \bibnamefont {Hamed~Mousavian}}, \bibinfo {author} {\bibfnamefont {D.}~\bibnamefont {Gottesman}},\ and\ \bibinfo {author} {\bibfnamefont {J.}~\bibnamefont {Emerson}},\ }\bibfield  {title} {\bibinfo {title} {The resource theory of stabilizer quantum computation},\ }\href {https://doi.org/10.1088/1367-2630/16/1/013009} {\bibfield  {journal} {\bibinfo  {journal} {New Journal of Physics}\ }\textbf {\bibinfo {volume} {16}},\ \bibinfo {pages} {013009} (\bibinfo {year} {2014})}\BibitemShut {NoStop}%
\bibitem [{\citenamefont {Mari}\ and\ \citenamefont {Eisert}(2012)}]{Mari2012}%
  \BibitemOpen
  \bibfield  {author} {\bibinfo {author} {\bibfnamefont {A.}~\bibnamefont {Mari}}\ and\ \bibinfo {author} {\bibfnamefont {J.}~\bibnamefont {Eisert}},\ }\bibfield  {title} {\bibinfo {title} {Positive wigner functions render classical simulation of quantum computation efficient},\ }\href {https://doi.org/10.1103/PhysRevLett.109.230503} {\bibfield  {journal} {\bibinfo  {journal} {Phys. Rev. Lett.}\ }\textbf {\bibinfo {volume} {109}},\ \bibinfo {pages} {230503} (\bibinfo {year} {2012})}\BibitemShut {NoStop}%
\bibitem [{\citenamefont {Raussendorf}\ \emph {et~al.}(2017)\citenamefont {Raussendorf}, \citenamefont {Browne}, \citenamefont {Delfosse}, \citenamefont {Okay},\ and\ \citenamefont {Bermejo-Vega}}]{Raussendorf2017}%
  \BibitemOpen
  \bibfield  {author} {\bibinfo {author} {\bibfnamefont {R.}~\bibnamefont {Raussendorf}}, \bibinfo {author} {\bibfnamefont {D.~E.}\ \bibnamefont {Browne}}, \bibinfo {author} {\bibfnamefont {N.}~\bibnamefont {Delfosse}}, \bibinfo {author} {\bibfnamefont {C.}~\bibnamefont {Okay}},\ and\ \bibinfo {author} {\bibfnamefont {J.}~\bibnamefont {Bermejo-Vega}},\ }\bibfield  {title} {\bibinfo {title} {Contextuality and wigner-function negativity in qubit quantum computation},\ }\href {https://doi.org/10.1103/PhysRevA.95.052334} {\bibfield  {journal} {\bibinfo  {journal} {Phys. Rev. A}\ }\textbf {\bibinfo {volume} {95}},\ \bibinfo {pages} {052334} (\bibinfo {year} {2017})}\BibitemShut {NoStop}%
\bibitem [{\citenamefont {Park}\ \emph {et~al.}(2024)\citenamefont {Park}, \citenamefont {Kwon},\ and\ \citenamefont {Jeong}}]{Park2024}%
  \BibitemOpen
  \bibfield  {author} {\bibinfo {author} {\bibfnamefont {G.}~\bibnamefont {Park}}, \bibinfo {author} {\bibfnamefont {H.}~\bibnamefont {Kwon}},\ and\ \bibinfo {author} {\bibfnamefont {H.}~\bibnamefont {Jeong}},\ }\bibfield  {title} {\bibinfo {title} {Extending classically simulatable bounds of clifford circuits with nonstabilizer states via framed wigner functions},\ }\bibfield  {journal} {\bibinfo  {journal} {Physical Review Letters}\ }\textbf {\bibinfo {volume} {133}},\ \href {https://doi.org/10.1103/physrevlett.133.220601} {10.1103/physrevlett.133.220601} (\bibinfo {year} {2024})\BibitemShut {NoStop}%
\bibitem [{\citenamefont {Raussendorf}\ \emph {et~al.}(2020)\citenamefont {Raussendorf}, \citenamefont {Bermejo-Vega}, \citenamefont {Tyhurst}, \citenamefont {Okay},\ and\ \citenamefont {Zurel}}]{Raussendorf2020}%
  \BibitemOpen
  \bibfield  {author} {\bibinfo {author} {\bibfnamefont {R.}~\bibnamefont {Raussendorf}}, \bibinfo {author} {\bibfnamefont {J.}~\bibnamefont {Bermejo-Vega}}, \bibinfo {author} {\bibfnamefont {E.}~\bibnamefont {Tyhurst}}, \bibinfo {author} {\bibfnamefont {C.}~\bibnamefont {Okay}},\ and\ \bibinfo {author} {\bibfnamefont {M.}~\bibnamefont {Zurel}},\ }\bibfield  {title} {\bibinfo {title} {Phase-space-simulation method for quantum computation with magic states on qubits},\ }\href {https://doi.org/10.1103/PhysRevA.101.012350} {\bibfield  {journal} {\bibinfo  {journal} {Phys. Rev. A}\ }\textbf {\bibinfo {volume} {101}},\ \bibinfo {pages} {012350} (\bibinfo {year} {2020})}\BibitemShut {NoStop}%
\bibitem [{\citenamefont {Okay}\ \emph {et~al.}(2021)\citenamefont {Okay}, \citenamefont {Zurel},\ and\ \citenamefont {Raussendorf}}]{Okay2021}%
  \BibitemOpen
  \bibfield  {author} {\bibinfo {author} {\bibfnamefont {C.}~\bibnamefont {Okay}}, \bibinfo {author} {\bibfnamefont {M.}~\bibnamefont {Zurel}},\ and\ \bibinfo {author} {\bibfnamefont {R.}~\bibnamefont {Raussendorf}},\ }\bibfield  {title} {\bibinfo {title} {On the extremal points of the lambda polytopes and classical simulation of quantum computation with magic states},\ }\href {https://doi.org/10.26421/qic21.13-14-2} {\bibfield  {journal} {\bibinfo  {journal} {Quantum Information and Computation}\ }\textbf {\bibinfo {volume} {21}},\ \bibinfo {pages} {1091–1110} (\bibinfo {year} {2021})}\BibitemShut {NoStop}%
\bibitem [{\citenamefont {Zurel}\ \emph {et~al.}(2020)\citenamefont {Zurel}, \citenamefont {Okay},\ and\ \citenamefont {Raussendorf}}]{Zurel2020}%
  \BibitemOpen
  \bibfield  {author} {\bibinfo {author} {\bibfnamefont {M.}~\bibnamefont {Zurel}}, \bibinfo {author} {\bibfnamefont {C.}~\bibnamefont {Okay}},\ and\ \bibinfo {author} {\bibfnamefont {R.}~\bibnamefont {Raussendorf}},\ }\bibfield  {title} {\bibinfo {title} {Hidden variable model for universal quantum computation with magic states on qubits},\ }\href {https://doi.org/10.1103/PhysRevLett.125.260404} {\bibfield  {journal} {\bibinfo  {journal} {Phys. Rev. Lett.}\ }\textbf {\bibinfo {volume} {125}},\ \bibinfo {pages} {260404} (\bibinfo {year} {2020})}\BibitemShut {NoStop}%
\bibitem [{\citenamefont {Zurel}\ \emph {et~al.}(2024)\citenamefont {Zurel}, \citenamefont {Okay},\ and\ \citenamefont {Raussendorf}}]{Zurel2024}%
  \BibitemOpen
  \bibfield  {author} {\bibinfo {author} {\bibfnamefont {M.}~\bibnamefont {Zurel}}, \bibinfo {author} {\bibfnamefont {C.}~\bibnamefont {Okay}},\ and\ \bibinfo {author} {\bibfnamefont {R.}~\bibnamefont {Raussendorf}},\ }\bibfield  {title} {\bibinfo {title} {Simulating quantum computation: How many “bits” for “it”?},\ }\bibfield  {journal} {\bibinfo  {journal} {PRX Quantum}\ }\textbf {\bibinfo {volume} {5}},\ \href {https://doi.org/10.1103/prxquantum.5.030343} {10.1103/prxquantum.5.030343} (\bibinfo {year} {2024})\BibitemShut {NoStop}%
\bibitem [{\citenamefont {Zurel}\ \emph {et~al.}(2025)\citenamefont {Zurel}, \citenamefont {Cohen},\ and\ \citenamefont {Raussendorf}}]{Zurel2025simulationquantumcomputationmagic}%
  \BibitemOpen
  \bibfield  {author} {\bibinfo {author} {\bibfnamefont {M.}~\bibnamefont {Zurel}}, \bibinfo {author} {\bibfnamefont {L.~Z.}\ \bibnamefont {Cohen}},\ and\ \bibinfo {author} {\bibfnamefont {R.}~\bibnamefont {Raussendorf}},\ }\href {https://arxiv.org/abs/2307.16034} {\bibinfo {title} {Simulation of quantum computation with magic states via jordan-wigner transformations}} (\bibinfo {year} {2025}),\ \Eprint {https://arxiv.org/abs/2307.16034} {arXiv:2307.16034 [quant-ph]} \BibitemShut {NoStop}%
\bibitem [{\citenamefont {Zurel}\ and\ \citenamefont {Heimendahl}(2024)}]{zurel2024efficientclassicalsimulationquantum}%
  \BibitemOpen
  \bibfield  {author} {\bibinfo {author} {\bibfnamefont {M.}~\bibnamefont {Zurel}}\ and\ \bibinfo {author} {\bibfnamefont {A.}~\bibnamefont {Heimendahl}},\ }\href {https://arxiv.org/abs/2407.10349} {\bibinfo {title} {Efficient classical simulation of quantum computation beyond wigner positivity}} (\bibinfo {year} {2024}),\ \Eprint {https://arxiv.org/abs/2407.10349} {arXiv:2407.10349 [quant-ph]} \BibitemShut {NoStop}%
\bibitem [{\citenamefont {Pashayan}\ \emph {et~al.}(2015)\citenamefont {Pashayan}, \citenamefont {Wallman},\ and\ \citenamefont {Bartlett}}]{Pashayan2015}%
  \BibitemOpen
  \bibfield  {author} {\bibinfo {author} {\bibfnamefont {H.}~\bibnamefont {Pashayan}}, \bibinfo {author} {\bibfnamefont {J.~J.}\ \bibnamefont {Wallman}},\ and\ \bibinfo {author} {\bibfnamefont {S.~D.}\ \bibnamefont {Bartlett}},\ }\bibfield  {title} {\bibinfo {title} {Estimating outcome probabilities of quantum circuits using quasiprobabilities},\ }\bibfield  {journal} {\bibinfo  {journal} {Physical Review Letters}\ }\textbf {\bibinfo {volume} {115}},\ \href {https://doi.org/10.1103/physrevlett.115.070501} {10.1103/physrevlett.115.070501} (\bibinfo {year} {2015})\BibitemShut {NoStop}%
\bibitem [{\citenamefont {Gross}(2006)}]{Gross2006}%
  \BibitemOpen
  \bibfield  {author} {\bibinfo {author} {\bibfnamefont {D.}~\bibnamefont {Gross}},\ }\bibfield  {title} {\bibinfo {title} {Hudson’s theorem for finite-dimensional quantum systems},\ }\bibfield  {journal} {\bibinfo  {journal} {Journal of Mathematical Physics}\ }\textbf {\bibinfo {volume} {47}},\ \href {https://doi.org/10.1063/1.2393152} {10.1063/1.2393152} (\bibinfo {year} {2006})\BibitemShut {NoStop}%
\bibitem [{foo()}]{footnote:GrossBoundStates}%
  \BibitemOpen
  \href@noop {} {}\bibinfo {note} {There exist magic states that are positive with respect to Gross' Wigner function. Therefore, standard phase-space simulation methods can be used to simulate such states. Hence, such states must be bound-magic states.}\BibitemShut {Stop}%
\bibitem [{\citenamefont {Schmid}\ \emph {et~al.}(2022)\citenamefont {Schmid}, \citenamefont {Du}, \citenamefont {Selby},\ and\ \citenamefont {Pusey}}]{SchmidWignerUniqueness}%
  \BibitemOpen
  \bibfield  {author} {\bibinfo {author} {\bibfnamefont {D.}~\bibnamefont {Schmid}}, \bibinfo {author} {\bibfnamefont {H.}~\bibnamefont {Du}}, \bibinfo {author} {\bibfnamefont {J.~H.}\ \bibnamefont {Selby}},\ and\ \bibinfo {author} {\bibfnamefont {M.~F.}\ \bibnamefont {Pusey}},\ }\bibfield  {title} {\bibinfo {title} {Uniqueness of noncontextual models for stabilizer subtheories},\ }\href {https://doi.org/10.1103/PhysRevLett.129.120403} {\bibfield  {journal} {\bibinfo  {journal} {Phys. Rev. Lett.}\ }\textbf {\bibinfo {volume} {129}},\ \bibinfo {pages} {120403} (\bibinfo {year} {2022})}\BibitemShut {NoStop}%
\bibitem [{\citenamefont {Zurel}(2024)}]{Zurel_2024}%
  \BibitemOpen
  \bibfield  {author} {\bibinfo {author} {\bibfnamefont {M.}~\bibnamefont {Zurel}},\ }\emph {\bibinfo {title} {Classical descriptions of quantum computations: foundations of quantum computation via hidden variable models, quasiprobability representations, and classical simulation algorithms}},\ \href {https://doi.org/http://dx.doi.org/10.14288/1.0445223} {Ph.D. thesis},\ \bibinfo  {school} {University of British Columbia} (\bibinfo {year} {2024})\BibitemShut {NoStop}%
\bibitem [{\citenamefont {Rudolph}\ and\ \citenamefont {Grover}(2002)}]{rudolph2002}%
  \BibitemOpen
  \bibfield  {author} {\bibinfo {author} {\bibfnamefont {T.}~\bibnamefont {Rudolph}}\ and\ \bibinfo {author} {\bibfnamefont {L.}~\bibnamefont {Grover}},\ }\href {https://arxiv.org/abs/quant-ph/0210187} {\bibinfo {title} {A 2 rebit gate universal for quantum computing}} (\bibinfo {year} {2002}),\ \Eprint {https://arxiv.org/abs/quant-ph/0210187} {arXiv:quant-ph/0210187 [quant-ph]} \BibitemShut {NoStop}%
\bibitem [{\citenamefont {Kirkwood}(1933)}]{Kirkwood33}%
  \BibitemOpen
  \bibfield  {author} {\bibinfo {author} {\bibfnamefont {J.~G.}\ \bibnamefont {Kirkwood}},\ }\bibfield  {title} {\bibinfo {title} {Quantum statistics of almost classical assemblies},\ }\href@noop {} {\bibfield  {journal} {\bibinfo  {journal} {Physical Review}\ }\textbf {\bibinfo {volume} {44}},\ \bibinfo {pages} {31} (\bibinfo {year} {1933})}\BibitemShut {NoStop}%
\bibitem [{\citenamefont {Dirac}(1945)}]{Dirac45}%
  \BibitemOpen
  \bibfield  {author} {\bibinfo {author} {\bibfnamefont {P.~A.~M.}\ \bibnamefont {Dirac}},\ }\bibfield  {title} {\bibinfo {title} {On the analogy between classical and quantum mechanics},\ }\href {https://doi.org/10.1103/RevModPhys.17.195} {\bibfield  {journal} {\bibinfo  {journal} {Rev. Mod. Phys.}\ }\textbf {\bibinfo {volume} {17}},\ \bibinfo {pages} {195} (\bibinfo {year} {1945})}\BibitemShut {NoStop}%
\bibitem [{\citenamefont {Lostaglio}\ \emph {et~al.}(2023{\natexlab{a}})\citenamefont {Lostaglio}, \citenamefont {Belenchia}, \citenamefont {Levy}, \citenamefont {Hern{\'{a}}ndez-G{\'{o}}mez}, \citenamefont {Fabbri},\ and\ \citenamefont {Gherardini}}]{Lostaglio2023kirkwooddirac}%
  \BibitemOpen
  \bibfield  {author} {\bibinfo {author} {\bibfnamefont {M.}~\bibnamefont {Lostaglio}}, \bibinfo {author} {\bibfnamefont {A.}~\bibnamefont {Belenchia}}, \bibinfo {author} {\bibfnamefont {A.}~\bibnamefont {Levy}}, \bibinfo {author} {\bibfnamefont {S.}~\bibnamefont {Hern{\'{a}}ndez-G{\'{o}}mez}}, \bibinfo {author} {\bibfnamefont {N.}~\bibnamefont {Fabbri}},\ and\ \bibinfo {author} {\bibfnamefont {S.}~\bibnamefont {Gherardini}},\ }\bibfield  {title} {\bibinfo {title} {Kirkwood-{D}irac quasiprobability approach to the statistics of incompatible observables},\ }\href {https://doi.org/10.22331/q-2023-10-09-1128} {\bibfield  {journal} {\bibinfo  {journal} {{Quantum}}\ }\textbf {\bibinfo {volume} {7}},\ \bibinfo {pages} {1128} (\bibinfo {year} {2023}{\natexlab{a}})}\BibitemShut {NoStop}%
\bibitem [{\citenamefont {Arvidsson-Shukur}\ \emph {et~al.}(2024)\citenamefont {Arvidsson-Shukur}, \citenamefont {Braasch~Jr}, \citenamefont {De~Bièvre}, \citenamefont {Dressel}, \citenamefont {Jordan}, \citenamefont {Langrenez}, \citenamefont {Lostaglio}, \citenamefont {Lundeen},\ and\ \citenamefont {Halpern}}]{ArvidssonShukur2024}%
  \BibitemOpen
  \bibfield  {author} {\bibinfo {author} {\bibfnamefont {D.~R.~M.}\ \bibnamefont {Arvidsson-Shukur}}, \bibinfo {author} {\bibfnamefont {W.~F.}\ \bibnamefont {Braasch~Jr}}, \bibinfo {author} {\bibfnamefont {S.}~\bibnamefont {De~Bièvre}}, \bibinfo {author} {\bibfnamefont {J.}~\bibnamefont {Dressel}}, \bibinfo {author} {\bibfnamefont {A.~N.}\ \bibnamefont {Jordan}}, \bibinfo {author} {\bibfnamefont {C.}~\bibnamefont {Langrenez}}, \bibinfo {author} {\bibfnamefont {M.}~\bibnamefont {Lostaglio}}, \bibinfo {author} {\bibfnamefont {J.~S.}\ \bibnamefont {Lundeen}},\ and\ \bibinfo {author} {\bibfnamefont {N.~Y.}\ \bibnamefont {Halpern}},\ }\bibfield  {title} {\bibinfo {title} {Properties and applications of the kirkwood–dirac distribution},\ }\href {https://doi.org/10.1088/1367-2630/ada05d} {\bibfield  {journal} {\bibinfo  {journal} {New Journal of Physics}\ }\textbf {\bibinfo {volume} {26}},\ \bibinfo {pages} {121201} (\bibinfo {year} {2024})}\BibitemShut {NoStop}%
\bibitem [{\citenamefont {Langrenez}\ \emph {et~al.}(2024{\natexlab{a}})\citenamefont {Langrenez}, \citenamefont {Arvidsson-Shukur},\ and\ \citenamefont {De~Bièvre}}]{Langrenez2024}%
  \BibitemOpen
  \bibfield  {author} {\bibinfo {author} {\bibfnamefont {C.}~\bibnamefont {Langrenez}}, \bibinfo {author} {\bibfnamefont {D.~R.~M.}\ \bibnamefont {Arvidsson-Shukur}},\ and\ \bibinfo {author} {\bibfnamefont {S.}~\bibnamefont {De~Bièvre}},\ }\bibfield  {title} {\bibinfo {title} {Characterizing the geometry of the kirkwood–dirac-positive states},\ }\bibfield  {journal} {\bibinfo  {journal} {Journal of Mathematical Physics}\ }\textbf {\bibinfo {volume} {65}},\ \href {https://doi.org/10.1063/5.0164672} {10.1063/5.0164672} (\bibinfo {year} {2024}{\natexlab{a}})\BibitemShut {NoStop}%
\bibitem [{\citenamefont {Langrenez}\ \emph {et~al.}(2025)\citenamefont {Langrenez}, \citenamefont {Arvidsson-Shukur},\ and\ \citenamefont {De~Bievre}}]{Langrenez2025Convex}%
  \BibitemOpen
  \bibfield  {author} {\bibinfo {author} {\bibfnamefont {C.}~\bibnamefont {Langrenez}}, \bibinfo {author} {\bibfnamefont {D.~R.~M.}\ \bibnamefont {Arvidsson-Shukur}},\ and\ \bibinfo {author} {\bibfnamefont {S.}~\bibnamefont {De~Bievre}},\ }\bibfield  {title} {\bibinfo {title} {Convex roofs witnessing kirkwood-dirac nonpositivity},\ }\href {http://iopscience.iop.org/article/10.1088/1751-8121/add58f} {\bibfield  {journal} {\bibinfo  {journal} {Journal of Physics A: Mathematical and Theoretical}\ } (\bibinfo {year} {2025})}\BibitemShut {NoStop}%
\bibitem [{\citenamefont {De~Bièvre}\ \emph {et~al.}(2025)\citenamefont {De~Bièvre}, \citenamefont {Langrenez},\ and\ \citenamefont {Radchenko}}]{debièvre2025}%
  \BibitemOpen
  \bibfield  {author} {\bibinfo {author} {\bibfnamefont {S.}~\bibnamefont {De~Bièvre}}, \bibinfo {author} {\bibfnamefont {C.}~\bibnamefont {Langrenez}},\ and\ \bibinfo {author} {\bibfnamefont {D.}~\bibnamefont {Radchenko}},\ }\bibfield  {title} {\bibinfo {title} {The kirkwood-dirac representation associated to the fourier transform for finite abelian groups: positivity},\ }\href {https://arxiv.org/abs/2501.12252} {\bibfield  {journal} {\bibinfo  {journal} {arXiv preprint}\ } (\bibinfo {year} {2025})},\ \Eprint {https://arxiv.org/abs/2501.12252} {arXiv:2501.12252 [quant-ph]} \BibitemShut {NoStop}%
\bibitem [{\citenamefont {Calderbank}\ \emph {et~al.}(1998)\citenamefont {Calderbank}, \citenamefont {Rains}, \citenamefont {Shor},\ and\ \citenamefont {Sloane}}]{Calderbank98}%
  \BibitemOpen
  \bibfield  {author} {\bibinfo {author} {\bibfnamefont {A.}~\bibnamefont {Calderbank}}, \bibinfo {author} {\bibfnamefont {E.}~\bibnamefont {Rains}}, \bibinfo {author} {\bibfnamefont {P.}~\bibnamefont {Shor}},\ and\ \bibinfo {author} {\bibfnamefont {N.}~\bibnamefont {Sloane}},\ }\bibfield  {title} {\bibinfo {title} {Quantum error correction via codes over gf(4)},\ }\href {https://doi.org/10.1109/18.681315} {\bibfield  {journal} {\bibinfo  {journal} {IEEE Transactions on Information Theory}\ }\textbf {\bibinfo {volume} {44}},\ \bibinfo {pages} {1369} (\bibinfo {year} {1998})}\BibitemShut {NoStop}%
\bibitem [{\citenamefont {Steane}(1996)}]{Steane1996}%
  \BibitemOpen
  \bibfield  {author} {\bibinfo {author} {\bibfnamefont {A.}~\bibnamefont {Steane}},\ }\bibfield  {title} {\bibinfo {title} {Multiple-particle interference and quantum error correction},\ }\href {https://doi.org/10.1098/rspa.1996.0136} {\bibfield  {journal} {\bibinfo  {journal} {Proceedings of the Royal Society of London. Series A: Mathematical, Physical and Engineering Sciences}\ }\textbf {\bibinfo {volume} {452}},\ \bibinfo {pages} {2551–2577} (\bibinfo {year} {1996})}\BibitemShut {NoStop}%
\bibitem [{\citenamefont {Gottesman}(1997)}]{gottesman1997}%
  \BibitemOpen
  \bibfield  {author} {\bibinfo {author} {\bibfnamefont {D.}~\bibnamefont {Gottesman}},\ }\href {https://arxiv.org/abs/quant-ph/9705052} {\bibinfo {title} {Stabilizer codes and quantum error correction}} (\bibinfo {year} {1997}),\ \Eprint {https://arxiv.org/abs/quant-ph/9705052} {arXiv:quant-ph/9705052 [quant-ph]} \BibitemShut {NoStop}%
\bibitem [{Sup()}]{SuppMat}%
  \BibitemOpen
  \href@noop {} {\bibinfo {title} {{See Supplemental Material at \url{https://doi.org/10.1103/x819-898d} for detailed proofs, which includes Ref.~\cite{zyczkowski2001induced}}}}\BibitemShut {NoStop}%
\bibitem [{\citenamefont {Raussendorf}\ and\ \citenamefont {Harrington}(2007)}]{Raussendorf2007}%
  \BibitemOpen
  \bibfield  {author} {\bibinfo {author} {\bibfnamefont {R.}~\bibnamefont {Raussendorf}}\ and\ \bibinfo {author} {\bibfnamefont {J.}~\bibnamefont {Harrington}},\ }\bibfield  {title} {\bibinfo {title} {Fault-tolerant quantum computation with high threshold in two dimensions},\ }\href {https://doi.org/10.1103/PhysRevLett.98.190504} {\bibfield  {journal} {\bibinfo  {journal} {Phys. Rev. Lett.}\ }\textbf {\bibinfo {volume} {98}},\ \bibinfo {pages} {190504} (\bibinfo {year} {2007})}\BibitemShut {NoStop}%
\bibitem [{\citenamefont {Krinner}\ \emph {et~al.}(2022)\citenamefont {Krinner}, \citenamefont {Lacroix}, \citenamefont {Remm}, \citenamefont {Di~Paolo}, \citenamefont {Genois}, \citenamefont {Leroux}, \citenamefont {Hellings}, \citenamefont {Lazar}, \citenamefont {Swiadek}, \citenamefont {Herrmann}, \citenamefont {Norris}, \citenamefont {Andersen}, \citenamefont {M{\"u}ller}, \citenamefont {Blais}, \citenamefont {Eichler},\ and\ \citenamefont {Wallraff}}]{ExperimentalSurfaceCodes1}%
  \BibitemOpen
  \bibfield  {author} {\bibinfo {author} {\bibfnamefont {S.}~\bibnamefont {Krinner}}, \bibinfo {author} {\bibfnamefont {N.}~\bibnamefont {Lacroix}}, \bibinfo {author} {\bibfnamefont {A.}~\bibnamefont {Remm}}, \bibinfo {author} {\bibfnamefont {A.}~\bibnamefont {Di~Paolo}}, \bibinfo {author} {\bibfnamefont {E.}~\bibnamefont {Genois}}, \bibinfo {author} {\bibfnamefont {C.}~\bibnamefont {Leroux}}, \bibinfo {author} {\bibfnamefont {C.}~\bibnamefont {Hellings}}, \bibinfo {author} {\bibfnamefont {S.}~\bibnamefont {Lazar}}, \bibinfo {author} {\bibfnamefont {F.}~\bibnamefont {Swiadek}}, \bibinfo {author} {\bibfnamefont {J.}~\bibnamefont {Herrmann}}, \bibinfo {author} {\bibfnamefont {G.~J.}\ \bibnamefont {Norris}}, \bibinfo {author} {\bibfnamefont {C.~K.}\ \bibnamefont {Andersen}}, \bibinfo {author} {\bibfnamefont {M.}~\bibnamefont {M{\"u}ller}}, \bibinfo {author} {\bibfnamefont {A.}~\bibnamefont {Blais}}, \bibinfo {author} {\bibfnamefont {C.}~\bibnamefont {Eichler}},\ and\ \bibinfo {author} {\bibfnamefont
  {A.}~\bibnamefont {Wallraff}},\ }\bibfield  {title} {\bibinfo {title} {Realizing repeated quantum error correction in a distance-three surface code},\ }\href@noop {} {\bibfield  {journal} {\bibinfo  {journal} {Nature}\ }\textbf {\bibinfo {volume} {605}},\ \bibinfo {pages} {669} (\bibinfo {year} {2022})}\BibitemShut {NoStop}%
\bibitem [{\citenamefont {Andersen}\ \emph {et~al.}(2020)\citenamefont {Andersen}, \citenamefont {Remm}, \citenamefont {Lazar}, \citenamefont {Krinner}, \citenamefont {Lacroix}, \citenamefont {Norris}, \citenamefont {Gabureac}, \citenamefont {Eichler},\ and\ \citenamefont {Wallraff}}]{ExperimentalSurfaceCodes2}%
  \BibitemOpen
  \bibfield  {author} {\bibinfo {author} {\bibfnamefont {C.~K.}\ \bibnamefont {Andersen}}, \bibinfo {author} {\bibfnamefont {A.}~\bibnamefont {Remm}}, \bibinfo {author} {\bibfnamefont {S.}~\bibnamefont {Lazar}}, \bibinfo {author} {\bibfnamefont {S.}~\bibnamefont {Krinner}}, \bibinfo {author} {\bibfnamefont {N.}~\bibnamefont {Lacroix}}, \bibinfo {author} {\bibfnamefont {G.~J.}\ \bibnamefont {Norris}}, \bibinfo {author} {\bibfnamefont {M.}~\bibnamefont {Gabureac}}, \bibinfo {author} {\bibfnamefont {C.}~\bibnamefont {Eichler}},\ and\ \bibinfo {author} {\bibfnamefont {A.}~\bibnamefont {Wallraff}},\ }\bibfield  {title} {\bibinfo {title} {Repeated quantum error detection in a surface code},\ }\href@noop {} {\bibfield  {journal} {\bibinfo  {journal} {Nat. Phys.}\ }\textbf {\bibinfo {volume} {16}},\ \bibinfo {pages} {875} (\bibinfo {year} {2020})}\BibitemShut {NoStop}%
\bibitem [{\citenamefont {{Google Quantum AI and Collaborators}}(2025)}]{ExperimentalSurfaceCodes3}%
  \BibitemOpen
  \bibfield  {author} {\bibinfo {author} {\bibnamefont {{Google Quantum AI and Collaborators}}},\ }\bibfield  {title} {\bibinfo {title} {Quantum error correction below the surface code threshold},\ }\href@noop {} {\bibfield  {journal} {\bibinfo  {journal} {Nature}\ }\textbf {\bibinfo {volume} {638}},\ \bibinfo {pages} {920} (\bibinfo {year} {2025})}\BibitemShut {NoStop}%
\bibitem [{\citenamefont {Bernstein}\ and\ \citenamefont {Vazirani}(1997)}]{Bernstein1997}%
  \BibitemOpen
  \bibfield  {author} {\bibinfo {author} {\bibfnamefont {E.}~\bibnamefont {Bernstein}}\ and\ \bibinfo {author} {\bibfnamefont {U.}~\bibnamefont {Vazirani}},\ }\bibfield  {title} {\bibinfo {title} {Quantum complexity theory},\ }\href {https://doi.org/10.1137/S0097539796300921} {\bibfield  {journal} {\bibinfo  {journal} {SIAM Journal on Computing}\ }\textbf {\bibinfo {volume} {26}},\ \bibinfo {pages} {1411} (\bibinfo {year} {1997})},\ \Eprint {https://arxiv.org/abs/https://doi.org/10.1137/S0097539796300921} {https://doi.org/10.1137/S0097539796300921} \BibitemShut {NoStop}%
\bibitem [{\citenamefont {HUSIMI}(1940)}]{Husimi1940}%
  \BibitemOpen
  \bibfield  {author} {\bibinfo {author} {\bibfnamefont {K.}~\bibnamefont {HUSIMI}},\ }\bibfield  {title} {\bibinfo {title} {Some formal properties of the density matrix},\ }\href {https://doi.org/10.11429/ppmsj1919.22.4_264} {\bibfield  {journal} {\bibinfo  {journal} {Proceedings of the Physico-Mathematical Society of Japan. 3rd Series}\ }\textbf {\bibinfo {volume} {22}},\ \bibinfo {pages} {264} (\bibinfo {year} {1940})}\BibitemShut {NoStop}%
\bibitem [{\citenamefont {Arvidsson-Shukur}\ \emph {et~al.}(2020)\citenamefont {Arvidsson-Shukur}, \citenamefont {Yunger~Halpern}, \citenamefont {Lepage}, \citenamefont {Lasek}, \citenamefont {Barnes},\ and\ \citenamefont {Lloyd}}]{Arvidsson-Shukur2020}%
  \BibitemOpen
  \bibfield  {author} {\bibinfo {author} {\bibfnamefont {D.~R.~M.}\ \bibnamefont {Arvidsson-Shukur}}, \bibinfo {author} {\bibfnamefont {N.}~\bibnamefont {Yunger~Halpern}}, \bibinfo {author} {\bibfnamefont {H.~V.}\ \bibnamefont {Lepage}}, \bibinfo {author} {\bibfnamefont {A.~A.}\ \bibnamefont {Lasek}}, \bibinfo {author} {\bibfnamefont {C.~H.~W.}\ \bibnamefont {Barnes}},\ and\ \bibinfo {author} {\bibfnamefont {S.}~\bibnamefont {Lloyd}},\ }\bibfield  {title} {\bibinfo {title} {Quantum advantage in postselected metrology},\ }\href {https://doi.org/10.1038/s41467-020-17559-w} {\bibfield  {journal} {\bibinfo  {journal} {Nat. Commun.}\ }\textbf {\bibinfo {volume} {11}},\ \bibinfo {pages} {3775} (\bibinfo {year} {2020})}\BibitemShut {NoStop}%
\bibitem [{\citenamefont {Lupu-Gladstein}\ \emph {et~al.}(2022)\citenamefont {Lupu-Gladstein}, \citenamefont {Yilmaz}, \citenamefont {Arvidsson-Shukur}, \citenamefont {Brodutch}, \citenamefont {Pang}, \citenamefont {Steinberg},\ and\ \citenamefont {Yunger~Halpern}}]{lupu2022negative}%
  \BibitemOpen
  \bibfield  {author} {\bibinfo {author} {\bibfnamefont {N.}~\bibnamefont {Lupu-Gladstein}}, \bibinfo {author} {\bibfnamefont {Y.~B.}\ \bibnamefont {Yilmaz}}, \bibinfo {author} {\bibfnamefont {D.~R.~M.}\ \bibnamefont {Arvidsson-Shukur}}, \bibinfo {author} {\bibfnamefont {A.}~\bibnamefont {Brodutch}}, \bibinfo {author} {\bibfnamefont {A.~O.}\ \bibnamefont {Pang}}, \bibinfo {author} {\bibfnamefont {A.~M.}\ \bibnamefont {Steinberg}},\ and\ \bibinfo {author} {\bibfnamefont {N.}~\bibnamefont {Yunger~Halpern}},\ }\bibfield  {title} {\bibinfo {title} {Negative quasiprobabilities enhance phase estimation in quantum-optics experiment},\ }\href@noop {} {\bibfield  {journal} {\bibinfo  {journal} {Physical Review Letters}\ }\textbf {\bibinfo {volume} {128}},\ \bibinfo {pages} {220504} (\bibinfo {year} {2022})}\BibitemShut {NoStop}%
\bibitem [{\citenamefont {Jenne}\ and\ \citenamefont {Arvidsson-Shukur}(2022)}]{Jenne22}%
  \BibitemOpen
  \bibfield  {author} {\bibinfo {author} {\bibfnamefont {J.~H.}\ \bibnamefont {Jenne}}\ and\ \bibinfo {author} {\bibfnamefont {D.~R.~M.}\ \bibnamefont {Arvidsson-Shukur}},\ }\bibfield  {title} {\bibinfo {title} {Unbounded and lossless compression of multiparameter quantum information},\ }\href {https://doi.org/10.1103/PhysRevA.106.042404} {\bibfield  {journal} {\bibinfo  {journal} {Phys. Rev. A}\ }\textbf {\bibinfo {volume} {106}},\ \bibinfo {pages} {042404} (\bibinfo {year} {2022})}\BibitemShut {NoStop}%
\bibitem [{\citenamefont {Aharonov}\ \emph {et~al.}(1988)\citenamefont {Aharonov}, \citenamefont {Albert},\ and\ \citenamefont {Vaidman}}]{Aharonov88}%
  \BibitemOpen
  \bibfield  {author} {\bibinfo {author} {\bibfnamefont {Y.}~\bibnamefont {Aharonov}}, \bibinfo {author} {\bibfnamefont {D.~Z.}\ \bibnamefont {Albert}},\ and\ \bibinfo {author} {\bibfnamefont {L.}~\bibnamefont {Vaidman}},\ }\bibfield  {title} {\bibinfo {title} {How the result of a measurement of a component of the spin of a spin-1/2 particle can turn out to be 100},\ }\href {https://doi.org/10.1103/PhysRevLett.60.1351} {\bibfield  {journal} {\bibinfo  {journal} {Phys. Rev. Lett.}\ }\textbf {\bibinfo {volume} {60}},\ \bibinfo {pages} {1351} (\bibinfo {year} {1988})}\BibitemShut {NoStop}%
\bibitem [{\citenamefont {Dressel}\ \emph {et~al.}(2014)\citenamefont {Dressel}, \citenamefont {Malik}, \citenamefont {Miatto}, \citenamefont {Jordan},\ and\ \citenamefont {Boyd}}]{Dressel14}%
  \BibitemOpen
  \bibfield  {author} {\bibinfo {author} {\bibfnamefont {J.}~\bibnamefont {Dressel}}, \bibinfo {author} {\bibfnamefont {M.}~\bibnamefont {Malik}}, \bibinfo {author} {\bibfnamefont {F.~M.}\ \bibnamefont {Miatto}}, \bibinfo {author} {\bibfnamefont {A.~N.}\ \bibnamefont {Jordan}},\ and\ \bibinfo {author} {\bibfnamefont {R.~W.}\ \bibnamefont {Boyd}},\ }\bibfield  {title} {\bibinfo {title} {Colloquium: Understanding quantum weak values: Basics and applications},\ }\href {https://doi.org/10.1103/RevModPhys.86.307} {\bibfield  {journal} {\bibinfo  {journal} {Rev. Mod. Phys.}\ }\textbf {\bibinfo {volume} {86}},\ \bibinfo {pages} {307} (\bibinfo {year} {2014})}\BibitemShut {NoStop}%
\bibitem [{\citenamefont {Pang}\ \emph {et~al.}(2014)\citenamefont {Pang}, \citenamefont {Dressel},\ and\ \citenamefont {Brun}}]{Pang14}%
  \BibitemOpen
  \bibfield  {author} {\bibinfo {author} {\bibfnamefont {S.}~\bibnamefont {Pang}}, \bibinfo {author} {\bibfnamefont {J.}~\bibnamefont {Dressel}},\ and\ \bibinfo {author} {\bibfnamefont {T.~A.}\ \bibnamefont {Brun}},\ }\bibfield  {title} {\bibinfo {title} {Entanglement-assisted weak value amplification},\ }\href {https://doi.org/10.1103/PhysRevLett.113.030401} {\bibfield  {journal} {\bibinfo  {journal} {Phys. Rev. Lett.}\ }\textbf {\bibinfo {volume} {113}},\ \bibinfo {pages} {030401} (\bibinfo {year} {2014})}\BibitemShut {NoStop}%
\bibitem [{\citenamefont {Pusey}(2014)}]{Pusey14}%
  \BibitemOpen
  \bibfield  {author} {\bibinfo {author} {\bibfnamefont {M.~F.}\ \bibnamefont {Pusey}},\ }\bibfield  {title} {\bibinfo {title} {Anomalous weak values are proofs of contextuality},\ }\href {https://doi.org/10.1103/PhysRevLett.113.200401} {\bibfield  {journal} {\bibinfo  {journal} {Phys. Rev. Lett.}\ }\textbf {\bibinfo {volume} {113}},\ \bibinfo {pages} {200401} (\bibinfo {year} {2014})}\BibitemShut {NoStop}%
\bibitem [{\citenamefont {Dressel}(2015)}]{Dressel15}%
  \BibitemOpen
  \bibfield  {author} {\bibinfo {author} {\bibfnamefont {J.}~\bibnamefont {Dressel}},\ }\bibfield  {title} {\bibinfo {title} {Weak values as interference phenomena},\ }\href {https://doi.org/10.1103/PhysRevA.91.032116} {\bibfield  {journal} {\bibinfo  {journal} {Phys. Rev. A}\ }\textbf {\bibinfo {volume} {91}},\ \bibinfo {pages} {032116} (\bibinfo {year} {2015})}\BibitemShut {NoStop}%
\bibitem [{\citenamefont {Pang}\ and\ \citenamefont {Brun}(2015)}]{Pang15}%
  \BibitemOpen
  \bibfield  {author} {\bibinfo {author} {\bibfnamefont {S.}~\bibnamefont {Pang}}\ and\ \bibinfo {author} {\bibfnamefont {T.~A.}\ \bibnamefont {Brun}},\ }\bibfield  {title} {\bibinfo {title} {Improving the precision of weak measurements by postselection measurement},\ }\href {https://doi.org/10.1103/PhysRevLett.115.120401} {\bibfield  {journal} {\bibinfo  {journal} {Phys. Rev. Lett.}\ }\textbf {\bibinfo {volume} {115}},\ \bibinfo {pages} {120401} (\bibinfo {year} {2015})}\BibitemShut {NoStop}%
\bibitem [{\citenamefont {Thio}\ \emph {et~al.}(2025)\citenamefont {Thio}, \citenamefont {Salmon}, \citenamefont {Barnes}, \citenamefont {De~Bi\`evre},\ and\ \citenamefont {Arvidsson-Shukur}}]{Thio2024}%
  \BibitemOpen
  \bibfield  {author} {\bibinfo {author} {\bibfnamefont {J.~J.}\ \bibnamefont {Thio}}, \bibinfo {author} {\bibfnamefont {W.}~\bibnamefont {Salmon}}, \bibinfo {author} {\bibfnamefont {C.~H.~W.}\ \bibnamefont {Barnes}}, \bibinfo {author} {\bibfnamefont {S.}~\bibnamefont {De~Bi\`evre}},\ and\ \bibinfo {author} {\bibfnamefont {D.~R.~M.}\ \bibnamefont {Arvidsson-Shukur}},\ }\bibfield  {title} {\bibinfo {title} {Verification of contextuality by noncontextual experiments},\ }\href {https://doi.org/10.1103/mzw5-t326} {\bibfield  {journal} {\bibinfo  {journal} {Phys. Rev. A}\ }\textbf {\bibinfo {volume} {112}},\ \bibinfo {pages} {062230} (\bibinfo {year} {2025})}\BibitemShut {NoStop}%
\bibitem [{\citenamefont {Yunger~Halpern}(2017)}]{YungHalp17}%
  \BibitemOpen
  \bibfield  {author} {\bibinfo {author} {\bibfnamefont {N.}~\bibnamefont {Yunger~Halpern}},\ }\bibfield  {title} {\bibinfo {title} {Jarzynski-like equality for the out-of-time-ordered correlator},\ }\href {https://doi.org/10.1103/PhysRevA.95.012120} {\bibfield  {journal} {\bibinfo  {journal} {Phys. Rev. A}\ }\textbf {\bibinfo {volume} {95}},\ \bibinfo {pages} {012120} (\bibinfo {year} {2017})}\BibitemShut {NoStop}%
\bibitem [{\citenamefont {Lostaglio}(2018)}]{Lostaglio18}%
  \BibitemOpen
  \bibfield  {author} {\bibinfo {author} {\bibfnamefont {M.}~\bibnamefont {Lostaglio}},\ }\bibfield  {title} {\bibinfo {title} {Quantum fluctuation theorems, contextuality, and work quasiprobabilities},\ }\href {https://doi.org/10.1103/PhysRevLett.120.040602} {\bibfield  {journal} {\bibinfo  {journal} {Phys. Rev. Lett.}\ }\textbf {\bibinfo {volume} {120}},\ \bibinfo {pages} {040602} (\bibinfo {year} {2018})}\BibitemShut {NoStop}%
\bibitem [{\citenamefont {Lostaglio}(2020)}]{Lostaglio20}%
  \BibitemOpen
  \bibfield  {author} {\bibinfo {author} {\bibfnamefont {M.}~\bibnamefont {Lostaglio}},\ }\bibfield  {title} {\bibinfo {title} {Certifying quantum signatures in thermodynamics and metrology via contextuality of quantum linear response},\ }\href {https://doi.org/10.1103/PhysRevLett.125.230603} {\bibfield  {journal} {\bibinfo  {journal} {Phys. Rev. Lett.}\ }\textbf {\bibinfo {volume} {125}},\ \bibinfo {pages} {230603} (\bibinfo {year} {2020})}\BibitemShut {NoStop}%
\bibitem [{\citenamefont {Levy}\ and\ \citenamefont {Lostaglio}(2020)}]{Levy20}%
  \BibitemOpen
  \bibfield  {author} {\bibinfo {author} {\bibfnamefont {A.}~\bibnamefont {Levy}}\ and\ \bibinfo {author} {\bibfnamefont {M.}~\bibnamefont {Lostaglio}},\ }\bibfield  {title} {\bibinfo {title} {Quasiprobability distribution for heat fluctuations in the quantum regime},\ }\href {https://doi.org/10.1103/PRXQuantum.1.010309} {\bibfield  {journal} {\bibinfo  {journal} {PRX Quantum}\ }\textbf {\bibinfo {volume} {1}},\ \bibinfo {pages} {010309} (\bibinfo {year} {2020})}\BibitemShut {NoStop}%
\bibitem [{\citenamefont {Lostaglio}\ \emph {et~al.}(2023{\natexlab{b}})\citenamefont {Lostaglio}, \citenamefont {Belenchia}, \citenamefont {Levy}, \citenamefont {Hern{\'a}ndez-G{\'o}mez}, \citenamefont {Fabbri},\ and\ \citenamefont {Gherardini}}]{lostaglio2023kirkwood}%
  \BibitemOpen
  \bibfield  {author} {\bibinfo {author} {\bibfnamefont {M.}~\bibnamefont {Lostaglio}}, \bibinfo {author} {\bibfnamefont {A.}~\bibnamefont {Belenchia}}, \bibinfo {author} {\bibfnamefont {A.}~\bibnamefont {Levy}}, \bibinfo {author} {\bibfnamefont {S.}~\bibnamefont {Hern{\'a}ndez-G{\'o}mez}}, \bibinfo {author} {\bibfnamefont {N.}~\bibnamefont {Fabbri}},\ and\ \bibinfo {author} {\bibfnamefont {S.}~\bibnamefont {Gherardini}},\ }\bibfield  {title} {\bibinfo {title} {Kirkwood-dirac quasiprobability approach to the statistics of incompatible observables},\ }\href@noop {} {\bibfield  {journal} {\bibinfo  {journal} {Quantum}\ }\textbf {\bibinfo {volume} {7}},\ \bibinfo {pages} {1128} (\bibinfo {year} {2023}{\natexlab{b}})}\BibitemShut {NoStop}%
\bibitem [{\citenamefont {Upadhyaya}\ \emph {et~al.}(2024)\citenamefont {Upadhyaya}, \citenamefont {Braasch}, \citenamefont {Landi},\ and\ \citenamefont {Halpern}}]{Upadhyaya2024}%
  \BibitemOpen
  \bibfield  {author} {\bibinfo {author} {\bibfnamefont {T.}~\bibnamefont {Upadhyaya}}, \bibinfo {author} {\bibfnamefont {W.~F.}\ \bibnamefont {Braasch}}, \bibinfo {author} {\bibfnamefont {G.~T.}\ \bibnamefont {Landi}},\ and\ \bibinfo {author} {\bibfnamefont {N.~Y.}\ \bibnamefont {Halpern}},\ }\bibfield  {title} {\bibinfo {title} {Non-abelian transport distinguishes three usually equivalent notions of entropy production},\ }\href {https://doi.org/10.1103/PRXQuantum.5.030355} {\bibfield  {journal} {\bibinfo  {journal} {PRX Quantum}\ }\textbf {\bibinfo {volume} {5}},\ \bibinfo {pages} {030355} (\bibinfo {year} {2024})}\BibitemShut {NoStop}%
\bibitem [{\citenamefont {Yunger~Halpern}\ \emph {et~al.}(2018)\citenamefont {Yunger~Halpern}, \citenamefont {Swingle},\ and\ \citenamefont {Dressel}}]{YungHalp18}%
  \BibitemOpen
  \bibfield  {author} {\bibinfo {author} {\bibfnamefont {N.}~\bibnamefont {Yunger~Halpern}}, \bibinfo {author} {\bibfnamefont {B.}~\bibnamefont {Swingle}},\ and\ \bibinfo {author} {\bibfnamefont {J.}~\bibnamefont {Dressel}},\ }\bibfield  {title} {\bibinfo {title} {Quasiprobability behind the out-of-time-ordered correlator},\ }\href {https://doi.org/10.1103/PhysRevA.97.042105} {\bibfield  {journal} {\bibinfo  {journal} {Phys. Rev. A}\ }\textbf {\bibinfo {volume} {97}},\ \bibinfo {pages} {042105} (\bibinfo {year} {2018})}\BibitemShut {NoStop}%
\bibitem [{\citenamefont {Mohseninia}\ \emph {et~al.}(2019)\citenamefont {Mohseninia}, \citenamefont {Alonso},\ and\ \citenamefont {Dressel}}]{Razieh19}%
  \BibitemOpen
  \bibfield  {author} {\bibinfo {author} {\bibfnamefont {R.}~\bibnamefont {Mohseninia}}, \bibinfo {author} {\bibfnamefont {J.~R.~G.}\ \bibnamefont {Alonso}},\ and\ \bibinfo {author} {\bibfnamefont {J.}~\bibnamefont {Dressel}},\ }\bibfield  {title} {\bibinfo {title} {Optimizing measurement strengths for qubit quasiprobabilities behind out-of-time-ordered correlators},\ }\href {https://doi.org/10.1103/PhysRevA.100.062336} {\bibfield  {journal} {\bibinfo  {journal} {Phys. Rev. A}\ }\textbf {\bibinfo {volume} {100}},\ \bibinfo {pages} {062336} (\bibinfo {year} {2019})}\BibitemShut {NoStop}%
\bibitem [{\citenamefont {Gonz\'alez~Alonso}\ \emph {et~al.}(2019)\citenamefont {Gonz\'alez~Alonso}, \citenamefont {Yunger~Halpern},\ and\ \citenamefont {Dressel}}]{YungHalp19}%
  \BibitemOpen
  \bibfield  {author} {\bibinfo {author} {\bibfnamefont {J.~R.}\ \bibnamefont {Gonz\'alez~Alonso}}, \bibinfo {author} {\bibfnamefont {N.}~\bibnamefont {Yunger~Halpern}},\ and\ \bibinfo {author} {\bibfnamefont {J.}~\bibnamefont {Dressel}},\ }\bibfield  {title} {\bibinfo {title} {Out-of-time-ordered-correlator quasiprobabilities robustly witness scrambling},\ }\href {https://doi.org/10.1103/PhysRevLett.122.040404} {\bibfield  {journal} {\bibinfo  {journal} {Phys. Rev. Lett.}\ }\textbf {\bibinfo {volume} {122}},\ \bibinfo {pages} {040404} (\bibinfo {year} {2019})}\BibitemShut {NoStop}%
\bibitem [{\citenamefont {Yunger~Halpern}\ \emph {et~al.}(2019)\citenamefont {Yunger~Halpern}, \citenamefont {Bartolotta},\ and\ \citenamefont {Pollack}}]{Yunger19-2}%
  \BibitemOpen
  \bibfield  {author} {\bibinfo {author} {\bibfnamefont {N.}~\bibnamefont {Yunger~Halpern}}, \bibinfo {author} {\bibfnamefont {A.}~\bibnamefont {Bartolotta}},\ and\ \bibinfo {author} {\bibfnamefont {J.}~\bibnamefont {Pollack}},\ }\bibfield  {title} {\bibinfo {title} {Entropic uncertainty relations for quantum information scrambling},\ }\href {https://doi.org/10.1038/s42005-019-0179-8} {\bibfield  {journal} {\bibinfo  {journal} {Communications Physics}\ }\textbf {\bibinfo {volume} {2}},\ \bibinfo {pages} {92} (\bibinfo {year} {2019})}\BibitemShut {NoStop}%
\bibitem [{\citenamefont {Gherardini}\ and\ \citenamefont {De~Chiara}(2024)}]{Gherardini24}%
  \BibitemOpen
  \bibfield  {author} {\bibinfo {author} {\bibfnamefont {S.}~\bibnamefont {Gherardini}}\ and\ \bibinfo {author} {\bibfnamefont {G.}~\bibnamefont {De~Chiara}},\ }\bibfield  {title} {\bibinfo {title} {Quasiprobabilities in quantum thermodynamics and many-body systems},\ }\href {https://doi.org/10.1103/PRXQuantum.5.030201} {\bibfield  {journal} {\bibinfo  {journal} {PRX Quantum}\ }\textbf {\bibinfo {volume} {5}},\ \bibinfo {pages} {030201} (\bibinfo {year} {2024})}\BibitemShut {NoStop}%
\bibitem [{End()}]{Endnote1}%
  \BibitemOpen
  \href@noop {} {}\bibinfo {note} {The real part of a KD distribution is sometimes called a Terletsky–Margenau-Hill distribution \cite{Margenau61, Terletsky}.}\BibitemShut {Stop}%
\bibitem [{\citenamefont {Hudson}(1974)}]{Hudson1974}%
  \BibitemOpen
  \bibfield  {author} {\bibinfo {author} {\bibfnamefont {R.}~\bibnamefont {Hudson}},\ }\bibfield  {title} {\bibinfo {title} {When is the wigner quasi-probability density non-negative?},\ }\href {https://doi.org/https://doi.org/10.1016/0034-4877(74)90007-X} {\bibfield  {journal} {\bibinfo  {journal} {Reports on Mathematical Physics}\ }\textbf {\bibinfo {volume} {6}},\ \bibinfo {pages} {249} (\bibinfo {year} {1974})}\BibitemShut {NoStop}%
\bibitem [{\citenamefont {Wigner}(1932)}]{Wigner1932}%
  \BibitemOpen
  \bibfield  {author} {\bibinfo {author} {\bibfnamefont {E.}~\bibnamefont {Wigner}},\ }\bibfield  {title} {\bibinfo {title} {On the quantum correction for thermodynamic equilibrium},\ }\href {https://doi.org/10.1103/PhysRev.40.749} {\bibfield  {journal} {\bibinfo  {journal} {Phys. Rev.}\ }\textbf {\bibinfo {volume} {40}},\ \bibinfo {pages} {749} (\bibinfo {year} {1932})}\BibitemShut {NoStop}%
\bibitem [{\citenamefont {Langrenez}\ \emph {et~al.}(2024{\natexlab{b}})\citenamefont {Langrenez}, \citenamefont {Salmon}, \citenamefont {De~Bièvre}, \citenamefont {Thio}, \citenamefont {Long},\ and\ \citenamefont {Arvidsson-Shukur}}]{langrenez2024setkirkwooddiracpositivestates}%
  \BibitemOpen
  \bibfield  {author} {\bibinfo {author} {\bibfnamefont {C.}~\bibnamefont {Langrenez}}, \bibinfo {author} {\bibfnamefont {W.}~\bibnamefont {Salmon}}, \bibinfo {author} {\bibfnamefont {S.}~\bibnamefont {De~Bièvre}}, \bibinfo {author} {\bibfnamefont {J.~J.}\ \bibnamefont {Thio}}, \bibinfo {author} {\bibfnamefont {C.~K.}\ \bibnamefont {Long}},\ and\ \bibinfo {author} {\bibfnamefont {D.~R.~M.}\ \bibnamefont {Arvidsson-Shukur}},\ }\href {https://arxiv.org/abs/2405.17557} {\bibinfo {title} {The set of kirkwood-dirac positive states is almost always minimal}} (\bibinfo {year} {2024}{\natexlab{b}}),\ \Eprint {https://arxiv.org/abs/2405.17557} {arXiv:2405.17557 [quant-ph]} \BibitemShut {NoStop}%
\bibitem [{\citenamefont {Fukuda}(1995)}]{Fukuda1995}%
  \BibitemOpen
  \bibfield  {author} {\bibinfo {author} {\bibfnamefont {K.}~\bibnamefont {Fukuda}},\ }\href@noop {} {\emph {\bibinfo {title} {cdd/cdd+ Reference Manual}}},\ \bibinfo {organization} {ETH Zurich} (\bibinfo {year} {1995}),\ \bibinfo {note} {available at \url{https://github.com/cddlib/cddlib}}\BibitemShut {NoStop}%
\bibitem [{NIS()}]{NISTDisclaimer}%
  \BibitemOpen
  \href@noop {} {\bibinfo {title} {{Certain software is identified in this paper to specify the experimental procedure adequately. Such identification is not intended to imply recommendation or endorsement of any product or service by NIST; nor is it intended to imply that the software identified is necessarily the best available for the purpose.}}}\BibitemShut {Stop}%
\bibitem [{\citenamefont {Thio}(2025)}]{Thi0_2025_KDQuantumComputation}%
  \BibitemOpen
  \bibfield  {author} {\bibinfo {author} {\bibfnamefont {J.~J.}\ \bibnamefont {Thio}},\ }\href@noop {} {\bibinfo {title} {{KDQuantumComputation}}},\ \bibinfo {howpublished} {\url{https://github.com/JJThi0/KDQuantumComputation}} (\bibinfo {year} {2025})\BibitemShut {NoStop}%
\bibitem [{\citenamefont {Audenaert}\ and\ \citenamefont {Plenio}(2005)}]{Audenaert2005}%
  \BibitemOpen
  \bibfield  {author} {\bibinfo {author} {\bibfnamefont {K.~M.~R.}\ \bibnamefont {Audenaert}}\ and\ \bibinfo {author} {\bibfnamefont {M.~B.}\ \bibnamefont {Plenio}},\ }\bibfield  {title} {\bibinfo {title} {Entanglement on mixed stabilizer states: normal forms and reduction procedures},\ }\href {https://doi.org/10.1088/1367-2630/7/1/170} {\bibfield  {journal} {\bibinfo  {journal} {New Journal of Physics}\ }\textbf {\bibinfo {volume} {7}},\ \bibinfo {pages} {170–170} (\bibinfo {year} {2005})}\BibitemShut {NoStop}%
\bibitem [{\citenamefont {Zyczkowski}\ and\ \citenamefont {Sommers}(2001)}]{zyczkowski2001induced}%
  \BibitemOpen
  \bibfield  {author} {\bibinfo {author} {\bibfnamefont {K.}~\bibnamefont {Zyczkowski}}\ and\ \bibinfo {author} {\bibfnamefont {H.-J.}\ \bibnamefont {Sommers}},\ }\bibfield  {title} {\bibinfo {title} {Induced measures in the space of mixed quantum states},\ }\href@noop {} {\bibfield  {journal} {\bibinfo  {journal} {Journal of Physics A: Mathematical and General}\ }\textbf {\bibinfo {volume} {34}},\ \bibinfo {pages} {7111} (\bibinfo {year} {2001})}\BibitemShut {NoStop}%
\bibitem [{\citenamefont {Margenau}\ and\ \citenamefont {Hill}(1961)}]{Margenau61}%
  \BibitemOpen
  \bibfield  {author} {\bibinfo {author} {\bibfnamefont {H.}~\bibnamefont {Margenau}}\ and\ \bibinfo {author} {\bibfnamefont {R.~N.}\ \bibnamefont {Hill}},\ }\bibfield  {title} {\bibinfo {title} {Correlation between measurements in quantum theory:},\ }\href {https://doi.org/10.1143/PTP.26.722} {\bibfield  {journal} {\bibinfo  {journal} {Progress of Theoretical Physics}\ }\textbf {\bibinfo {volume} {26}},\ \bibinfo {pages} {722} (\bibinfo {year} {1961})},\ \Eprint {https://arxiv.org/abs/https://academic.oup.com/ptp/article-pdf/26/5/722/5454875/26-5-722.pdf} {https://academic.oup.com/ptp/article-pdf/26/5/722/5454875/26-5-722.pdf} \BibitemShut {NoStop}%
\bibitem [{\citenamefont {Terletsky}(1937)}]{Terletsky}%
  \BibitemOpen
  \bibfield  {author} {\bibinfo {author} {\bibfnamefont {Y.~P.}\ \bibnamefont {Terletsky}},\ }\bibfield  {title} {\bibinfo {title} {On the classical limit of quantum mechanics},\ }\href@noop {} {\bibfield  {journal} {\bibinfo  {journal} {Zh. Eksp. Teor. Fiz.}\ }\textbf {\bibinfo {volume} {1290}} (\bibinfo {year} {1937})}\BibitemShut {NoStop}%
\end{thebibliography}
\end{document}